\documentclass[12pt, draftclsnofoot, onecolumn]{IEEEtran}
\usepackage{mathtools}
\usepackage{amssymb}
\usepackage{amsmath,color,amsthm}
\usepackage{breqn}
\usepackage{bm}
\usepackage{cite}
\usepackage{algpseudocode}
\usepackage{algorithm}
\usepackage{algorithmicx}
\usepackage[table]{xcolor}
\usepackage{subfig}
\usepackage[normalem]{ulem}
\usepackage{amsmath}
\usepackage{soul}
\usepackage{stackengine}
\algnewcommand\algorithmicforeach{\textbf{for each}}
\algdef{S}[FOR]{ForEach}[1]{\algorithmicforeach\ #1\ \algorithmicdo}
\usepackage{eqparbox}
\newdimen{\algindent}
\algnewcommand\LeftComment[2]{%
\hspace{#1\algindent}$\triangleright$ \eqparbox{COMMENT}{#2} \hfill %
}
\algnewcommand\LeftCommentNoTriangle[2]{%
\hspace{#1\algindent} \eqparbox{COMMENT}{#2} \hfill %
}
\usepackage{enumitem}
\usepackage{caption}
\usepackage{lipsum}
\usepackage{cuted}
\usepackage{enumerate}
\usepackage{booktabs} 
\usepackage{graphicx}
\usepackage{float}
\usepackage{mathrsfs}

\usepackage{amsmath,amsthm,amssymb}

\newcommand{\gd}[2]{\mathcal G ( #1, #2) }

\newtheorem{lemma}{Lemma}

\usepackage{wasysym} 

\DeclarePairedDelimiterX\Basics[1](){ #1}

\usepackage{hyperref}
\makeatletter
\newcommand{\labeltarget}[1]{\Hy@raisedlink{\hypertarget{#1}{}}}
\begin{document}
\title{Distributed Beamforming Techniques for Cell-Free Wireless Networks Using Deep Reinforcement Learning} 
\author{Firas Fredj, Yasser Al-Eryani, Setareh Maghsudi, Mohamed Akrout, \\and Ekram Hossain, {\it Fellow, IEEE} \thanks{F. Fredj, Y. Al-Eryani, M. Akrout, and E. Hossain are with the Department of Electrical and Computer Engineering at the University of Manitoba, Canada (emails: fredjf1@myumanitoba.ca, aleryany@myumanitoba.ca, akroutm@myumanitoba.ca, Ekram.Hossain@umanitoba.ca). S. Maghsudi is with the Department of Computer Science, University of T\"{u}bingen, Germany (email: setareh.maghsudi@uni-tuebingen.de). This work was supported in part by a Discovery Grant from the Natural Sciences and Engineering Research Council of Canada (NSERC) and in part by Grant 01IS20051 from the German Federal Ministry of Education and Research (BMBF).}
}
\maketitle
%\se{To see comments, please click on "Review" on the top-right. Thanks.}
%\se{Please use \textcolor{red}{red} to make any changes in the paper. Thanks.}
%-------------------------------------------------Abstarct
\begin{abstract}
In a cell-free network, a large number of mobile devices are served simultaneously by several base stations (BSs)/access points(APs) using the same time/frequency resources. However, this creates high signal processing demands (e.g. for beamforming) at the transmitters and receivers. In this work, we develop centralized and distributed deep reinforcement learning (DRL)-based methods to optimize beamforming at the uplink of a cell-free network. First, we propose a fully centralized uplink beamforming method (i.e. centralized learning) that uses the Deep Deterministic Policy Gradient algorithm (DDPG) for an  offline-trained DRL model. We then enhance this method, in terms of convergence and performance, by using distributed experiences collected from different APs based on the Distributed Distributional Deterministic Policy Gradients algorithm (D4PG) in which the APs represent the distributed agents of the DRL model. To reduce the complexity of signal processing at the central processing unit (CPU), we propose a fully distributed DRL-based uplink beamforming scheme. This scheme divides the beamforming computations among distributed APs. {The proposed schemes are then benchmarked against two common linear beamforming schemes, namely, minimum mean square estimation (MMSE) and the simplified conjugate symmetric schemes.}
The results show that the D4PG scheme with distributed experience achieves the best performance irrespective of the network size. Furthermore, although the proposed distributed beamforming technique reduces the complexity of centralized learning in the DDPG algorithm, it performs better than the DDPG algorithm only for small-scale networks. The performance superiority of the fully centralized DDPG model becomes more evident as the number of APs and/or UEs increases. 
%Moreover, during the operation stage, all DRL models demonstrate a significantly shorter processing time than that of the conventional gradient ascent (GA) solution.  Compared to the other two methods, the fully distributed DRL scheme would reduce the computational (and hence hardware) complexity significantly. {Finally, by evaluating the trained DRL models on environments with changing CSI realizations, we observe the performance-complexity tradeoff for the centralized DRL-based  methods and the distributed DRL-based method.} 
The codes for all of our DRL implementations are available at \texttt{\url{https://github.com/RayRedd/Distributed\_beamforming\_rl}}
\end{abstract}  
%----------------------------------------------------Keywords
\begin{IEEEkeywords}
Cell-free network, beamforming, successive interference cancellation, deep reinforcement learning (DRL), deep deterministic policy gradient algorithm (DDPG), distributed distributional deterministic policy gradients algorithm (D4PG). 
\end{IEEEkeywords}
%===========================================Section======================================================
\section{Introduction}
\label{sec:Intro}

\subsection{Background}

To provide ultra-reliable low-latency communications (URLLC)~\cite{Massive_MTC_2} in the beyond 5G wireless systems (B5G), the idea of cell-free networks has emerged~\cite{Cell_Less_1}. A fully centralized cell-free network will use a fully connected wireless network architecture with centralized processing, control, and storage of data. Precisely, in a cell-free wireless network ({Fig.~\ref{System_Model}}), all access points (APs)/base stations (BSs) cooperate to simultaneously serve all user equipments (UEs) within the network coverage area \cite{Cell_Less_3,Cell_Less_4,DOMA}.
Such centralized network operations mitigate the adverse effects of non-coordinated collisions and interference among transmitted signals, especially in scenarios such as massive machine-type communications (mMTC) \cite{Cell_Less_2}. Furthermore, in a cell-free architecture, fast fronthaul/backhaul links connect all APs to an edge cloud processor which is responsible for  simultaneous downlink (uplink) beamforming design for transmit(receive) signals to(from) different UEs \cite{Cell_Less_Beamforming_1,Cell_Less_Beamforming_2}. On the downside, a fully centralized cell-free network architecture requires a huge computational capacity. Moreover, without an efficient design, there would be excessive control signaling \cite{Cell_Free_Complexity_1,Cell_Free_Complexity_2}. Note that the main benefits of fully centralized cell-free networks include enhanced coverage, improved diversity, and provisioning of efficient interference cancellation mechanisms.    

\begin{figure}[htb]
\centering
\includegraphics[scale=0.5]{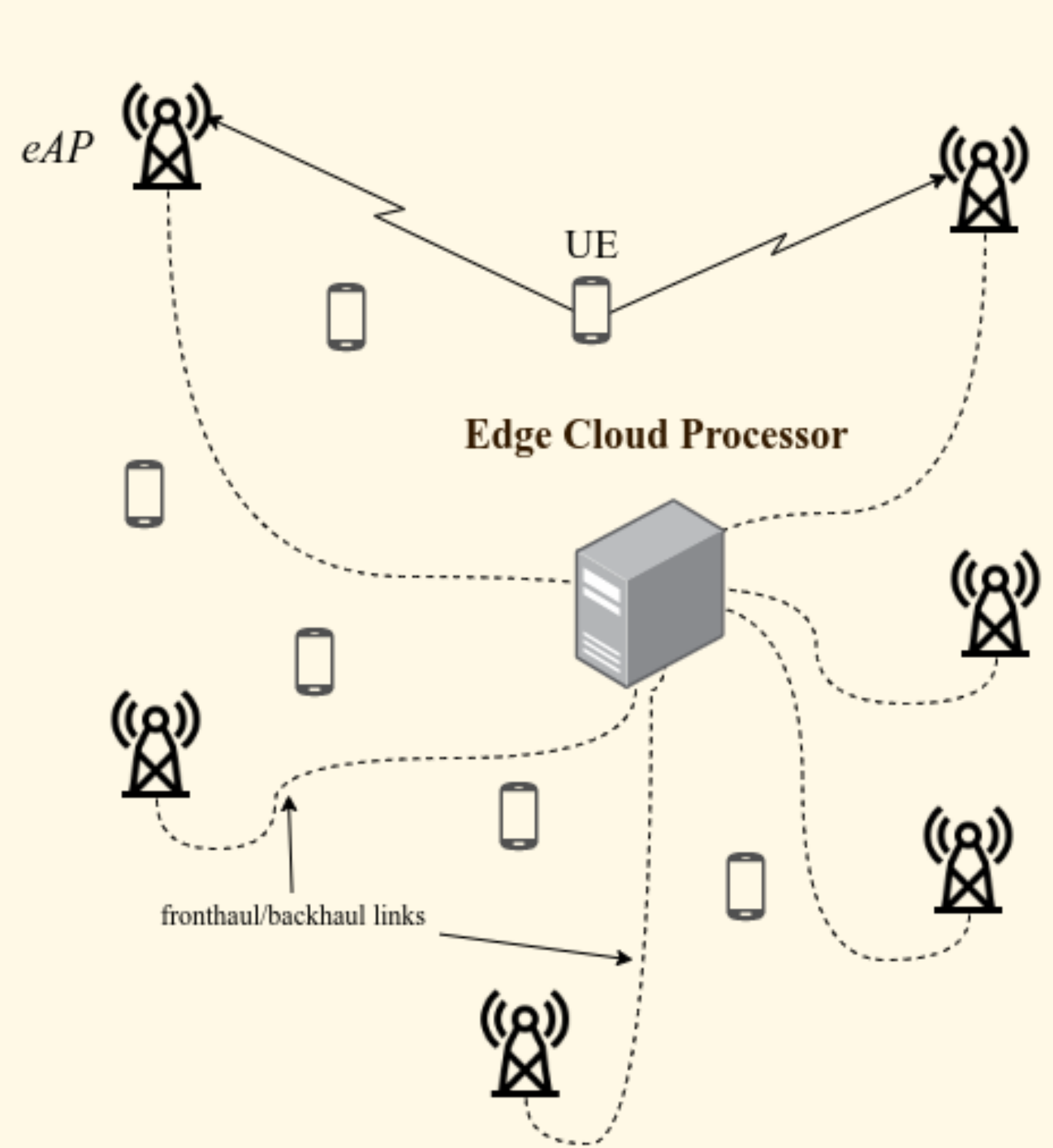}
\caption{A cell-free network model.}
\label{System_Model}
\end{figure}

The challenges associated with cell-free networking such as multi-UE joint beamforming and channel estimation \cite{DRL_MEC_1,DRL_MEC_2,DRL_MEC_3,DRL_MEC_4} can be addressed by using artificial intelligence (AI), specifically, machine learning (ML) techniques. Due to the computational complexity, these problems are often characterized by an {\it algorithmic deficit} rather than a {\it modeling deficit} \cite{simeone18}. 
Among numerous methods, deep reinforcement learning (DRL) is a notable candidate to design cell-free networks that avoids having a training data sets {\it a priori} which are  hard to obtain for dynamic wireless environments. Also, DRL enables us to achieve a trade-off between centralized- and distributed processing of computational tasks. Recent results have shown that, for a cell-free network,  simultaneous uplink/downlink beamforming within a centralized processing unit (CPU)  in optimal system performance. On the downside, fully centralized processing suffers from high computational complexity and excessive processing delay, especially when many UEs utilize the same time-frequency resources. However, as the signal processing tasks of an uplink/downlink cell-free network becomes more distributed, the network performance becomes closer to that of a cellular network with non-cooperative APs. 

%--------------Subsection: Related Works
%\subsection{Related Work}
%\label{sub:relatedWork}
%
The state-of-the-art of cell-free networks focuses on uplink/downlink beamforming \cite{Cell_Less_3, Cell_Less_Beamforming_2}, estimation of channel state information (CSI) \cite{DRL_MEC_1,Cell_Free_CSI_Est_2}, fronthaul imperfections \cite{Fronthaul_1}, and scalable cell-free network designs \cite{GCoMP, Cell_Free_Complexity_1, Cell_Free_Yasser}. 
For example, in \cite{Beamforming_Algorithms_1}, the authors propose conjugate beamforming and zero-forcing precoding scheme for a fully centralized downlink cell-free network. They show that the zero-forcing technique outperforms the conjugate beamforming technique. In \cite{DRL_MEC_1}, the authors develop a channel estimation technique for mmWave-enabled massive cell-free network using supervised learning-based denoising convolutional neural network. To reduce the complexity of centralized signal processing, \cite{GCoMP} proposes a partitioned cell-free wireless network architecture. The architecture clusters the cooperating APs based on current network CSI (UE-centric clustering). The scheme enables an efficient design of practical mMTC systems by compensating the effect of inter-cluster interference. The compensation is followed by network partitioning and enabling multi-level successive interference cancellation (SIC) at each receiver \cite{DOMA}. Another low-complexity design of cell-free network architecture is presented in \cite{Cell_Free_Yasser}. The core idea is to reduce the dimensionality of beamforming matrices by using a dynamic clustering of APs. Each cluster then represents a single multi-antenna AP (transmit/receive diversity). In \cite{DRL_MEC_3}, the authors utilize supervised learning to solve the beamforming problem in cell-free networks. They include a complete neural network optimizer in each AP. Every AP then obtains the local CSI knowledge by estimating only the large-scale fading while considering the small-scale fading as a constant. {Table \ref{SoA}} summarizes some of the important works in the area of beamforming in cell-free networks. Here, the system model column refers to the operational confguration  of the APs in a cell-free network (i.e. how the APs are grouped) as described below:
\begin{itemize}
    \item \textbf{Centralized:} All of the APs jointly serve the users and the processing is done centrally in the CPU, which is referred to as the edge cloud processor (ECP).
    
    \item \textbf{Co-located:} Each AP operates independently using its own processing capabilities.
    
    \item \textbf{UE-centric clustering:} Each user is assigned to cluster/group of APs.
    
    \item \textbf{Local partial zero forcing:} The channel estimation is performed locally at the APs.
    
    \item \textbf{Static:} The APs are considered in fixed locations in the network.
    
    \item \textbf{Dynamic:} The APs are grouped into clusters and the clusters are updated at each time
    step.
\end{itemize}

%---------------
\begin{table*}\scriptsize
\centering
\caption{Summary of beamforming schemes in cell-free networks}
\begin{tabular}{|c|c|c|c|}
\hline 
\textbf{Ref.} & \textbf{System model} & \textbf{Main objective} & \textbf{Techniques and characteristics} \\
\hline
\hline
\shortstack{\cite{Optimal_Cell_Free}\\ \textcolor{white}{.}} &  \shortstack{Centralized\\ \textcolor{white}{.}}  & \shortstack{Max-min fairness\\ \textcolor{white}{.}} & \shortstack{Formulates downlink beamforming problem as a quasi-concave\\ optimization and uses bisection method}  \\
\hline
 \shortstack{\cite{Angel_of_Arrival_Beamforming}\\ \textcolor{white}{.}} &  \shortstack{Centralized\\ \textcolor{white}{.}}  & \shortstack{Max-min fairness\\ \textcolor{white}{.}} & \shortstack{Designs an angle-of-arrival-based beamforming/combining\\  scheme for FDD-based cell-free network}  \\   
\hline
\shortstack{\cite{Cell_Less_3, Cell_Less_Beamforming_2,Conjugate_BF,ZF_VS_CB}\\ \textcolor{white}{.}\\ \textcolor{white}{.}}
 &
 \shortstack{ Centralized\\ \\ \textcolor{white}{.}}&
  \shortstack{ Max-min fairness\\ \\ \textcolor{white}{.}}
 &
\shortstack{Use conjugate beamforming and zero forcing techniques. \\ MMSE processing is also used for (partially-) \\ distributed operations of cell-free massive MIMO in \cite{Cell_Less_3}}\\
\hline
\shortstack{\cite{Schedualing_Cell_Free}  \textcolor{white}{ .}\\ \textcolor{white}{.}\\ \textcolor{white}{.}} & \shortstack{Co-located,\\ Cell-free \\ \textcolor{white}{.} } & \shortstack{Max-min fairness \\ \textcolor{white}{.} \\ \textcolor{white}{.}} & \shortstack{Uses Lyapunov optimization techniques to develop a dynamic \\ scheduling algorithm to perform user equipment (UE) scheduling based on \\ time slot and transmission rate} \\
\hline
\shortstack{\cite{GCoMP,Hybrid_Analog_Digital} \\ \textcolor{white}{.}} & \shortstack{UE-centric \\ clustering}  & \shortstack{ Maximize \\  per-cluster sum-rate} & \shortstack{First paper solves beamforming problem with optimal CSI, \\ Second paper proposes hybrid analog-digital beamforming} \\
\hline
\shortstack{\cite{Local_Precoding_Cell_Free}\\ \textcolor{white}{.}}  & \shortstack{{Local partial} \\ {zero-forcing}} & \shortstack{{Max-min fairness} \\ \textcolor{white}{ }} &
{\shortstack{Uses local CSI of the strongest connected UEs at each AP \\ to conduct local zero-forcing beamforming}} \\
\hline
\shortstack{\cite{ZF_RZF}\\ \textcolor{white}{.}}  & \shortstack{Static \\ \textcolor{white}{.}} & \shortstack{Maximize sum-rate \\ (uplink/downlink)} & \shortstack{Uses minimum-mean-square-error receiver in uplink, \\ uses zero-forcing and regularized zero-forcing beamforming in downlink}\\
\hline
\shortstack{\cite{Cell_Free_Yasser}\\ \textcolor{white}{.}} & \shortstack{Dynamic  \\ \textcolor{white}{.}} & \shortstack{Maximize sum-rate \\ \textcolor{white}{.}} & \shortstack{Uses a hybrid DRL-based model for joint AP clustering and \\  beamforming that utilizes the DDQN and DDPG algorithms}\\
\hline
\shortstack{\textbf{This paper}\\ \textcolor{white}{.} \\ \textcolor{white}{.}\\ \textcolor{white}{.}} & \shortstack{Static\\ \textcolor{white}{.} \\ \textcolor{white}{.} \\ \textcolor{white}{.}}  & \shortstack{Maximize sum-rate \\ \textcolor{white}{.}\\ \textcolor{white}{.}\\ \textcolor{white}{.}} & \shortstack{The paper uses DDPG and D4PG algorithms \\  for beamforming based on (i) centralized learning, \\ (ii) centralized learning with distributed experiences, \\
and (iii) completely distributed learning.}\\
\hline

\end{tabular}
\label{SoA}
\end{table*}
%The use of DRL methods in wireless optimization problems 
%===============SubSection======================================================
\subsection{{Motivation and Contributions}}
\label{eq:Motivation}
{A fully centralized cell-free architecture  enhances the network coverage and transmission performance. The complexity of signal processing can be reduced by using scalable, dynamic, and low-complexity designs \cite{Cell_Free_Yasser}, and also by using suitable network- and UE-centric architectures and algorithms \cite{GCoMP}, \cite{Cell_Free_Scalability_1,Cell_Less_Beamforming_1}. 
The aforementioned low-complexity designs however sacrifice the performance gain of centralized processing. The complexity of solving the beamforming problem in a centralized manner (e.g. to obtain the beamforming vectors at the ECP) can however be reduced by using a distributed learning or processing approach while the detection of the  transmitted data is still performed at the central unit. Such a solution has not been investigated in the literature.} 

Several recent research papers have employed deep learning (DL) techniques to find approximately optimal solution of the beamforming problem in the uplink of cell-free networks based on centralized learning  \cite{Cell_free_Learning,DNN_Cell_Free}. However, these papers use a data-driven supervised-learning approach, which requires a large amount of training data for many different wireless propagation scenarios. Such a requirement is a severe limitation for practical implementation of a cell-free network. 
Furthermore, all previous DL-based beamforming techniques utilize the concept of deep Q-learning (DQN)  (or deep convolutional neural
network (DCNN)) with a \textit{discrete} action space. However, DQN is a RL algorithm that learns the state-value function and is designed to solve tasks with discrete action space. One way of using DQN on continuous action space is discretization. However, such a method suffers from the curse of dimensionality because the action space would grow exponentially with the size of the network. This results in sub-optimal beamforming designs. Therefore, in this paper, we propose to utilize DRL methods based on DDPG algorithm that handles \textit{continuous} action space for uplink beamforming in a cell-free network.  Specifically, in this paper, we investigate several practical beamforming designs for uplink cell-free network considering both centralized and distributed learning settings.  The main contributions of this paper can be summarized as follows:
\begin{itemize}
\item For a fully centralized uplink cell-free network, we formulate the optimal beamforming  problem in order to maximize the normalized sum-rate of the network. We then use a centralized learning approach to realize the uplink beamforming by using the deep deterministic policy gradient (DDPG) algorithm \cite{lillicrap2015continuous}.
\item We also propose a novel distributed experience-based beamforming system using the distributed distributional deterministic policy gradient (D4PG) algorithm \cite{barth2018distributed} in which the APs act as the distributed agents.  
\item To reduce the complexity of centralized learning, we propose a novel DRL-based beamforming scheme with distributed learning.
\item We evaluate the performance of the proposed beamforming designs numerically for different system settings considering non-orthogonal pilot contamination and shadow fading.

\item We compare the performances of the proposed DRL methods with those of the traditional conjugate method and MMSE method for beamforming. 
\end{itemize}  
The rest of this paper is organized as follows. The system model and the beamforming problem are presented in Section \ref{sec:SystemModel}. Section \ref{sec: Priliminaries} briefly reviews the  preliminaries of DRL and DRL-based beamforming. We design a centralized DRL solution, a distributed experienced-based  DRL solution, and another beamforming solution based on distributed learning for the beamforming problem in Section \ref{sec:CentralBeam}.  In Section \ref{sec:DistLearn}, we discuss the complexity of all of the DRL-based solutions. In Section \ref{sec:Numerical}, we present and discuss the numerical results. Section \ref{sec:Conclusion} concludes the paper. Definitions of the major system model parameters, DRL models parameters and the abbreviations used in the paper are given in Table \ref{table:symbol}.

%------------------------Notations
\noindent
\textbf{Notations:} For a random variable (rv) $X$, $F_X(x)$ and $f_X(x)$, respectively, represent cumulative distribution function (CDF) and probability density function (PDF). Moreover, $\mathbb E[\cdot ]$ denotes the expectation. For a given matrix $\bm{A}\in \mathbb{C}^{M\times N}$, $\bm{A}^H$ represents the Hermitian transpose of $\bm{A}$. The PDF of a random variable $X$ following the Nakagami-${m}$ distribution is given by $f_X(x)=\frac{2{m}^{{m}}}{\Gamma({m}) {\Omega}^{{m}}}x^{2{m}-1}e^{-\frac{{m}}{\Omega}x^2}$. A random variable $X$ that follows the Gamma distribution is denoted by $X \thicksim \mathcal G (\alpha, \beta)$, with the PDF being $f_{X}(x) = \frac{\beta^\alpha }{\Gamma(\alpha)} x^{\alpha -1} e^{- \beta x }, \quad x >0$, where $ \beta >0$, $ \alpha \geq  1 $, and $ \Gamma(z)$ is the Euler's Gamma function. Moreover, the distributions $O\gd {\alpha_i} {\beta_i}$, $i=1, \dots, N$, refer to the $i$-th ascending-ordered rv from a set of $N$  gamma rvs with parameters $\alpha_i$ and $\beta_i$. 
%===========================================Section======================================================
\section{System Model, Assumptions, and Problem Formulation}
\label{sec:SystemModel}
%===============SubSection======================================================
\subsection{Cell-free Network Model}
\label{subsec:Network}
We consider the uplink of a wireless network with $M$ single-antenna APs and $K$ single-antenna UEs that have fixed locations within a certain coverage area, as shown in \textbf{Fig.~\ref{System_Model}}. Each AP has a baseband processor to partially process the signal received from all connected UEs. We refer to such an AP as an ``\textit{enhanced-AP}" (eAP) to distinguish them from the conventional APs. The eAPs are connected via the backhaul links, hence forming a cell-free network \cite{Cell_Less_1}. Such a  network architecture enables the distributed eAPs to collaboratively serve all the UEs within the network coverage area. The beamforming optimization can be done either in a edge cloud processor (ECP) or in the eAPs in a distributed manner (Fig.~1). The eAPs are connected to the ECP in a star network topology. All of the eAPs use the same available spectrum to serve the same set of UEs. That is, the frequency reuse factor is  1 and all of the eAPs are allocated the same total channel bandwidth.
%To obtain the CSI, a random set of orthogonal pilot sequences is assigned to each UE. 
 
\begin{table}[htb] 
	\scriptsize \caption{Definitions of major system model parameters, DRL model parameters and abbreviations}
	\label{table:symbol}
	\begin{center}
		\renewcommand{\arraystretch}{1.0}
		\setlength\tabcolsep{3pt}
		\begin{tabular}{p{4cm}p{7cm}}
			\hline
			Parameter,  & Definition    \\
			Abbreviation & \\
			\hline 
			
\textbf{System model parameters} & \\
IUI & Intra-UE Interference \\
rv & random variable\\
$M$ & Number of eAPs \\
$K$ & Number of UEs \\
eAP & Enhanced Access Point \\
ECP & Edge Cloud Processor \\
UE & User Equipment\\
$\left(\alpha,\beta\right)$ & parameters of gamma rv\\
$\mathcal{F}_{mk}$  & Large-scale fading parameter\\
$\kappa$ & Path-loss exponent\\
$\tau_c$ & Channel coherence time\\
$\tau_p$ & Duration of pilot symbol\\
$\bm{\phi}_k$ & Pilot symbol of UE $k$\\
$\mathcal{E}_{mk}$ & Pilot estimation factor at link $k\rightarrow m$\\
$\rho_k$ & $k$-th Pilot transmission power\\
$p_k$ & $k$-th UE transmission power\\
$\eta_z$ & Additive White Gaussian Noise (AWGN) at point $z$ \\
$\rho_z$ & AWGN variance at point $z$ \\
$\bm{W}$ & Beamforming matrix\\
$\bm{\omega}_m$ & $m$-th row of the beamforming matrix $\mathbf{W}$ \\
\hline
\textbf{DRL parameters} & \\

$Q(s,a)$ & Action-value function (DDPG critic network) \\
$\mu(s)$ & Policy function (DDPG \& D4PG actor network) \\
$Z(s,a)$ & Action-value distribution function (D4PG critic network) \\
$\mathcal{S}$ & State space \\
$s$ & Sample state of DRL environment \\
$\mathcal{A}$ & Action space \\
$a$  & Sample action of actor $\mu(.)$\\
$r$ & Reward of DRL environment \\
$\zeta$ & Discount factor of DRL model \\
$\bm{\theta}^Q$ & Weights of $Q(.,.)$ \\
$\bm{\theta}^\mu$ &  Weights of $\mu(.)$\\
$\bm{\theta}^Z$ & Weights of $Z(.,.)$ \\

 \hline
		\end{tabular}
	\end{center}\label{Symboles}
\end{table}   

\subsection{Channel Model}
Between the $k\text{-th}$ UE and the $m\text{-th}$ eAP, the channel gain is a random variable 
\begin{equation}
\label{eq:Channel_Model}
g_{mk}=\mathcal{F}^{1/2}_{mk}h_{mk}.
\end{equation}
In (\ref{eq:Channel_Model}), $h_{mk}$ is the small-scale channel fading gain that follows a Nakagami-$\mathcal{M}_{mk}$ distribution with spreading and shape parameters $\mathcal{M}_{mk}$ and $\Omega_{mk}$, respectively. Therefore, $|h_{mk}|^2 \thicksim \mathcal{G}\left(\alpha_{mk}, \beta_{mk}\right)$, where $\alpha_{mk} = \mathcal{M}_{mk}$ is the shape parameter and $\beta_{mk} = \frac{\mathcal{M}_{mk}}{\Omega_{mk}}$ is the inverse-scale parameter. Moreover, we have
\begin{equation}
\mathcal{F}_{mk} = {L}_{mk}^{-2\kappa}10^{\frac{\sigma_{\text{sh}}z_{mk}}{10}},
\end{equation}
where $L_{mk}=||d_{mk}||$ is the Euclidean distance between the $k\text{-th}$ UE and the $m$-th eAP. Also, $\kappa$ is the path-loss exponent ($\kappa\geq 2$) and $\sigma_{\text{sh}}$ is the shadow fading variance in dB. Furthermore, $z_{mk}=\sqrt{\delta}a_{mk}+\sqrt{1-\delta}b_{mk}$, where $0\leq \delta \leq 1$ is the transmitter-receiver shadow fading correlation coefficient \cite{CellLessShadow}. The parameters $a_{mk}\thicksim \mathcal{N} \left(0,1\right)$ and $b_{mk}\thicksim \mathcal{N} \left(0,1\right)$ characterize the shadow fading. We assume that $L_{mk}^{-2\kappa}$ and $\mathcal{F}_{mk}~ \forall~m$, $k$ are known. This assumption is justified since large-scale fading parameters can be easily estimated given the received signal strength (RSS).
Accordingly, we have $|g_{mk}|^2\thicksim \mathcal{G}(\alpha_{mk},\beta_{mk}/\mathcal{F}_{mk})$.
%===============SubSection======================================================
\subsection{Uplink Network Training}
\label{subsec:UplinkNet}
To estimate the CSI of the cell-free network, we train the network using a set of orthonormal pilot sequences. \textbf{Lemma \ref{Lemma_1}} presents the MMSE estimation constant of the channel gain $g_{mk}$ in the presented cell-free network. 
% ===============Lemma Begin===================
\begin{lemma}\label{Lemma_1}
Let ${\bm{\varphi}}_{k}=\left[\varphi_{k, 1} \dots \varphi_{k, \tau_p} \right]^H$, $||{\bm{\varphi}}_{k}||^2=1$, be the pilot sequence with sample size $\tau_p$ that is assigned to the $k$-th UE. The MMSE estimation constant $\mathcal{E}_{mk}$ of $g_{mk}$ is given by
\begin{equation}  
\mathcal{E}_{mk}=\frac{\sqrt{\tau_p\rho_k}\mathcal{F}_{mk} \left(\frac{\alpha_{mk}}{\beta_{mk}}\right)}{\tau_p\sum_{l = 1}^K\rho_l\mathcal{F}_{ml}\left(\frac{\alpha_{ml}}{\beta_{ml}}\right)|\bm{\varphi}_{mk}^H\bm{\varphi}_{ml}|^2 + 1}.
\label{Training_4}   
\end{equation} 
\end{lemma}
\begin{proof}
See \textbf{Appendix A}.   
\end{proof}
% ================= Lemma End =====================
%Let ${\bm{\varphi}}_{k}=\left[\varphi_{k, 1} \dots \varphi_{k, \tau_p} \right]^H$, $||{\bm{\varphi}}_{k}||^2=1$, be the pilot sequence with sample size $\tau_p$ that is assigned to the $k$-th UE.

%the MMSE estimation constant $\mathcal{E}_{mk}$ yields
%
%\begin{equation}
%\mathcal{E}_{mk}=\frac{\sqrt{\tau_p\rho_k}\mathcal{F}_{mk} \left(\frac{\alpha_{mk}}{\beta_{mk}}\right)}{\tau_p\sum_{l = 1}^K\rho_l\mathcal{F}_{ml}\left(\frac{\alpha_{ml}}{\beta_{ml}}\right)|\bm{\varphi}_{mk}^H\bm{\varphi}_{ml}|^2 + 1},
%\label{Training_4}
%\end{equation}
%
In the above, we use the fact that $\mathbb{E}\left[|h_{mk}|^2\right] = \frac{\alpha_{mk}}{\beta_{mk}}$. If all of the UEs receive a set of mutually orthogonal pilot sequences (i.e. $|\bm{\varphi}^H_{k} \bm{\varphi}_{l}|=0, ~\forall~ k\neq l$), the estimated small-scale channel fading in (\ref{Training_2}) reduces to a scaled version of the exact fading gain plus a relatively small AWGN.
However, depending on the applications and due to the limitations concerning the length of the training sequence, non-orthogonal pilot signals have to be used among some active UEs. To decrease the computational complexity at the ECP, one approach would be to estimate the CSI at distributed eAPs, where the $m$-th eAP estimates channel gains $g_{mk}, \forall~k=1, \dots, K$.

%===============SubSection======================================================
\subsection{Uplink Data Transmission} \label{subsec:UplinkData}
In a cell-free network, each eAP receives the composite of signals from all UEs. For each UE, a weighted sum of composite signals from all eAPs is constructed to maximize the signal component while minimizing the interference plus noise. This process takes place at the baseband level in the ECP, before forwarding the detected signal of each UE to its final destination. Formally, the overall signal received by the ECP to be used in detecting the $k$-th UE's data is given by (\ref{Received_Signal_2}), where $w_{mk}$ is the $m$-th element of the beamforming vector related to the $k$-th UE such that $0\leq w_{mk}\leq 1$.
\begin{figure*}
\begin{equation}
\begin{split}
y_k  &=\sum_{m = 1}^{M} w_{mk} \left[\sum_{l=1}^K\hat{g}_{ml}\sqrt{p_l}x_l+{\tilde{\eta}}_{m}\right] \\
  &=  \underbrace{\sqrt{\tau_p\rho_k p_k}x_k\sum_{m = 1}^{M}w_{mk}\mathcal{E}_{mk}g_{mk}}_{\text{Desired Signal}} +  \sum_{m =1}^{M}w_{mk}\left[\underbrace{\sum_{l=1, l\neq k}^K \sqrt{\tau_p\rho_lp_l}x_l\mathcal{E}_{ml}g_{ml}}_{\text{Inter-UE Interference}}\right.  \\
&\;\;\;\;+\left. \underbrace{\sum_{v=1, v\neq k}^K\sqrt{\tau_p\rho_vp_k}x_k\mathcal{E}_{mk}|\bm{\varphi}_{k}^H \bm{\varphi}_{v}|g_{mv}+\sum_{q=1, q\neq k}^K \sum_{u=1, u\neq q}^K\sqrt{\tau_p\rho_{u}p_q}x_q\mathcal{E}_{mq}|\bm{\varphi}_{q}^H\bm{\varphi}_{u}|g_{mu}}_{\text{non-orthogonal pilot-related estimation error}}\right]  \\
  &\;\;\;\;+\sum_{m=1}^{M}w_{mk}\left[\underbrace{\sqrt{p_k}\mathcal{E}_{mk}x_k|\bm{\varphi}_{k}^H\bm{\eta}_{m}|+\sum_{z=1, z\neq k}^K\sqrt{p_z}\mathcal{E}_{mz}x_z|\bm{\varphi}_{z}^H\bm{\eta}_{m}|} _{\text{AWGN-related estimation error}}+\underbrace{\tilde{\eta}_{m}}_{\text{AWGN}}\right]. \label{Received_Signal_2} 
\end{split}
\end{equation}
\end{figure*}

Moreover, $p_{k}$ is the uplink transmission power of the $k$-th UE such that $0 \leq p_k \leq P_k$, where $P_k$ is the power budget of the $k$-th UE. Also, $x_k$ is the transmitted symbol of the $k$-th UE such that $\mathbb{E}[|x_k|^2]=1$, and $\tilde{\eta}_{m}$ is the AWGN at the $m$-th eAP with $\tilde{\eta}_{m}\thicksim \mathcal{CN}\left(0, 1/2\right)$. The instantaneous signal-to-interference-plus-noise ratio (SINR) for the $k$-th UE is given by (\ref{SINR_1}) \cite{SINR_Expression}
%
%\begin{dmath}
\begin{figure*}
\begin{equation}
\gamma_k=\frac{\sum_{m=1}^Mw_{mk}^2|\tilde{g}_{mk}|^2}{\sum_{m=1}^Mw_{mk}^2\left[\sum_{l=1, l\neq k}^K|\tilde{g}_{ml}|^2+ \sum_{v=1, v\neq k}^K|\tilde{g}_{m{v}}|^2+\sum_{q=1, q\neq k}^K\sum_{u=1, u\neq q}^K|\tilde{g}_{m{u}}|^2\right]+1},
\label{SINR_1}
\end{equation}
%\end{dmath}
\end{figure*}  
where $|\tilde{g}_{mi}|^2\thicksim \mathcal{G}\left(\tilde{\alpha}_{mi},\tilde{\beta}_{mi} \right)$, for $i=k,l,v,u$. Moreover, $\tilde{\alpha}_{mi}=\mathcal{M}_{mi}$, for $i=k,l,v,u$. In addition, $\tilde{\beta}_{mi}=\frac{\mathcal{M}_{mi}\dot{\sigma}_{mi}}{\Omega_{mi}\mathcal{F}_{mi} \tau_p\rho_ip_i\mathcal{E}_{mi}^2}$, for $i=k,l$. Also, $\tilde{\beta}_{mv}=\frac{\mathcal{M}_{mk}\dot{\sigma}_{mk} }{\Omega_{mk}\mathcal{F}_{mvs}\tau_p\rho_vp_k\mathcal{E}_{mk}^2|\bm{\varphi}_{k}^H\bm{\varphi}_{v}|^2}$. Similarly,   $\tilde{\beta}_{mu}=\frac{\mathcal{M}_{mq}\dot{\sigma}_{mk}}{\Omega_{mq}\mathcal{F}_{mu}\tau_p\rho_{u}p_{q}\mathcal{E}_{mq}^2|\bm{\varphi}_{q}^H\bm{\varphi}_{u}|^2}$, with

\begin{equation}
\dot{\sigma}_{mk} = \sum_{m=1}^M w_{mk}^2 \left[ \sum_{t=1}^{\tau_p} \left(p_k\mathcal{E}_{mk}^2 \varphi_{k,t}^2  + \sum_{z=1, z\neq k}^K~p_z\mathcal{E}_{mz}^2\varphi_{z,t}^2\right)+ 1\right]. \notag
\end{equation}
%jj
  
Equation (\ref{SINR_1}) is concluded from the following: When both the transmitter and the receiver know the estimated CSI, one can replace the second moments of channel fading parameters with their instantaneous values. For example, the numerator of (\ref{SINR_1}) can be written as
$\mathbb{E}\left[|\bm{\mathcal{G}}_k^H\bm{\mathcal{W}}_kx_k|^2\right]=\bm{\mathcal{W}}_k^H \bm{\mathcal{R}}_k\bm{\mathcal{W}}_k$,
where $\bm{\mathcal{G}}_k=\left[C_{1k}g_{1k}\dots C_Mg_{\mathcal{D}_1k} \right]$, $\bm{\mathcal{W}}_k=\left[ w_1 \dots w_M \right]$, and
$C_{mk}=\sqrt{\tau_p\rho_kp_k}\mathcal{E}_{mk}G_{mk}$. Moreover, $\bm{\mathcal{R}}_k$ represents the auto-correlation matrix of $k$-th UE's signal and is defined as $\bm{\mathcal{R}}_k=\mathbb{E}\left[|\bm{\mathcal{G}}_k^H\bm{\mathcal{G}}_k|^2\right]= \bar{\bm{\mathcal{G}}}_k\bar{\bm{\mathcal{G}}}_k^H + {\bm{C}}_{\bm{\mathcal{G}}_k}$,
where $\bar{\bm{\mathcal{G}}}_k$ and ${\bm{C}}_{\bm{\mathcal{G}}_k}$ are the mean and covariance matrices of $\bm{\mathcal{G}}_k$, respectively. If both the transmitter and the receiver know the instantaneous CSI, $\bm{\mathcal{R}}_k$ yields $\bm{\mathcal{R}}_k = \bm{\mathcal{G}}_k \bm{\mathcal{G}}_k^H$ \cite{SINR_Expression}. Using a similar procedure, one can characterize the interference power component of (\ref{SINR_1}). Additionally, one can compute the power of AWGN component by utilizing the fact that all noise samples are i.i.d circularly symmetric Gaussian rvs with zero-mean and constant variance $\sigma_{m}^2=\sigma^2=1/2, ~\forall~ m$. 
%===============SubSection======================================================
\subsection{Problem Formulation}
\label{subsec:GPF}
Since all UEs transmit in the same time-frequency channel, the receiver can deploy successive interference cancellation (SIC) to increase the UEs' SINR values \cite{8491054}, thereby enhancing the per-UE performance. To this end, the beamforming vector of each UE is designed optimally to separate the signal component of that UE from other components at least by some value $\bar{P_s}$ referred to as {\it receiver sensitivity}. More precisely, to decode UE $k$, the receiver first decodes the signal components of other UEs with higher power. It subtracts these components from the overall signal. It then decodes the desired signal treating the remaining UEs' signals, i.e. those with lower power, as interference. Formally, when detecting the signal from the $k$-th UE, 
we first arrange the received signal components from the UEs in an ascending order such that $\sum_{m = 1}^{M}~|\check{g}_{m1}|^2 \leq \dots \leq \sum_{m = 1}^{M} |\check{g}_{mk}|^2~~\leq\dots\leq\sum_{m = 1}^{M} |\check{g}_{mK}|^2$ \cite{GCoMP}. The beamforming vector of the $k$-th UE (denoted by $\bm{w}_{k}=\left[w_{1k}\dots w_{Mk}\right]$) shall maximize an objective function which is a function of $\gamma_k, ~\forall k=1, \dots, K$. For this, $\gamma_k$ in (\ref{SINR_1}) is modified as in (\ref{SINR_3}),
%
%\begin{dmath}
%\text{\hspace{-3mm}} 
\begin{figure*}
\begin{equation}
\gamma_k=\frac{\sum_{m = 1}^{M}w_{mk}^2|\check{g}_{mk}|^2}{\sum_{m = 1}^{M}w_{mk}^2\left[\sum_{l=k+1}^{K}|\check{g}_{ml}|^2+\sum_{v=1, v\neq k}^K|\check{g}_{{m}{v}}|^2+\sum_{q=1, q\neq k}^K\sum_{u=1, u\neq q}^K|\check{g}_{{m}{u}}|^2\right]+1}.
\label{SINR_3}
\end{equation}
\end{figure*}
%\end{dmath}
%
in which $|\check{g}_{mi}|^2 \thicksim O\gd {\check{\alpha}_{mi}} {\check{\beta}_{mi}}$ and $\check{\alpha}_{mi}=\mathcal{M}_{mi}$ for $i=k,l,v,u$. Moreover, $\check{\beta}_{mi}=\frac{ \mathcal{M}_{mi}\dot{\sigma}_{mk}}{\Omega_{mi}\mathcal{F}_{mi}\tau_p \rho_i p_i \mathcal{E}_{mi}^2}$ for $i=k,l$. In addition, $\check{\beta}_{mv}=\frac{\mathcal{M}_{mk}\dot{\sigma}_{mk}}{\Omega_{mk}\mathcal{F}_{mv}  \tau_p\rho_vp_k\mathcal{E}_{mk}^2|\bm{\varphi}_{k}^H\bm{\varphi}_{v}|^2}$, and $\check{\beta}_{mu}=\frac{\mathcal{M}_{mq}\dot{\sigma}_{mk}}{\Omega_{mq}\mathcal{F}_{mu}\tau_p\rho_{u}p_{q}\mathcal{E}_{mq}^2|\bm{\varphi}_{q}^H\bm{\varphi}_{u}|^2}$. Also, we have
 \begin{dmath}
 \dot{\sigma}_{mk}=\sum_{m = 1}^{M} w_{mk}^2\left[\sum_{t=1}^{\tau_p}\left(p_k\mathcal{E}_{mk}^2\varphi_{k,t}
 +\sum_{l=1, l\neq k}^Kp_l\mathcal{E}_{ml}^2\varphi_{l,t}\right)+ 1\right]. \notag
\end{dmath}       
The beamforming problem can be then formulated as 
\begin{equation} 
\begin{aligned}
 \textbf{P}_1:& ~~\underset{ \bm{W}\in [0~ 1]^{M \times K}}{\text{maximize}}~
%\underset{k=1, \dots K}{\text{min}}
 \sum_{k=1}^K\log_2\left(1+\gamma_{k}  \right)\\
%& ~\text{Subject to:} \\
\text{s.t.}~&~\textbf{C}_1: \sum_{m=1}^{M} \left(w_{m \delta_l}^2-\sum_{i=\delta_l+1}^{l}w_{mi}^2 \right)\bar{\gamma}_{{m}l}\geq \bar{P_s}, \\
&  ~  \textbf{C}_2:w_{mk}^2\in \left[0, 1\right],~~\forall~m, k,
\\
&~~\forall k,  \forall~\delta_l= 1, \dots, l-1\;\mbox{, and}\; l=2, \dots, K,\\
\end{aligned}
\label{OptimizationProb}
\end{equation}
where $\bar{\gamma}_{{m}l}=p_l| {g}_{ml}|^2$, and $\bm{W}\in \left[0,1 \right]^{M \times K}$ is the overall beamforming matrix in which $\bm{w}_k=[w_{1k}, \dots, w_{Mk}]$. Note that in (\ref{OptimizationProb}), the objective  is a function of $\gamma_k$, which is a function of $\bm{\omega}_k$. The constraint $\textbf{C}_1$ corresponds to the $\sum_{l=2}^K(l-1)=\frac{K(K-1)}{2}$ conditions of successful SIC operation with a receiver sensitivity of $\bar{P_s}$. 
Note that, for maximizing the minimum transmission rate (max-min fairness), the objective function of the optimization problem will be $\underset{k=1,...,K}{\text{minimum}} \log_2\left(1+\gamma_{k} \right)$. The receiver deals only with the measured (estimated) channel values that include the estimation error and the AWGN component. The SINR value after the SIC operation decreases due to  pilot contamination. 

The globally optimal solution of problem $\textbf{P}_1$ in (\ref{OptimizationProb}) is the one that gives the best performance among all possible matrices $\bm{W}$. However, $\textbf{P}_1$ is a non-convex problem due to the non-convexity of the objective function. Accordingly, the globally optimal solution of  $\textbf{P}_1$ in (\ref{OptimizationProb}) can only be obtained by an exhaustive search. The computational complexity of such as exhaustive search will be combinatorial in terms of the number of UEs and eAPs, and therefore, will not be feasible for practical implementation.
Iterative optimization methods can be used, for example, at the ECP of the cell-free system, that can converge to locally optimal solutions after a finite number of iterations. However, this will require all the CSI information to be transmitted to the ECP resulting in increased network overhead as well as high processing complexity. DRL-based solutions based on distributed learning can reduce the network overhead as well as processing complexity at the ECP. More specifically, centralized learning with distributed agents' experience can improve the quality of the solution, while a completely distributed learning-based scheme will reduce the signal processing complexity at the ECP and the network overhead significantly.
In the later sections of this paper, we will develop three DRL-based solutions, namely, the  fully centralized training and learning-based method, the centralized learning-based method with distributed training, and the fully distributed training and learning-based method, for the formulated problem. The centralized learning-based method can serve as a benchmark for the two other methods.
%===========================================Section======================================================
\section{Background on DRL and Beamforming Optimization}
\label{sec: Priliminaries}
%===============SubSection======================================================
\subsection{DRL Preliminaries}
%In reinforcement learning (RL), an agent interacts with an environment in discrete time. At each step, the agent takes an action and receives some reward while observing a new state or some other information. The agent utilizes the feedback to learn a behavior (policy) that maximizes the cumulative reward. 

DRL is a solution method for Markov decision processes (MDPs). An MDP is characterized by a tuple $(\mathcal{S}, \mathcal{A}, \mathcal{P}, \mathcal{R}, \zeta$), where $\mathcal{S}$ and $\mathcal{A}$ represent the state space and the agent's action space, respectively. 
Moreover, $\mathcal{P}$ is the transition probability matrix, where $\mathcal{P}(s,a,s') \in [0,1]$ is the probability that state $s$ changes to state $s'$ by selecting action $a$. $\mathcal{R}: \mathcal{S} \times \mathcal{A} \rightarrow \mathbb{R}$ defines the expected reward of performing action $a$ at state $s$. 
Finally, $\zeta$ is the reward's discount factor. The goal is to select the best action at each step so as to maximize the accumulated discounted reward. In a DRL model, a neural network learns to map the states to values or state-action to Q-values using the historical outcomes. DRL algorithms fall into three categories: (i) \textit{value-based} methods that aim to learn a value function like deep Q-learning (DQN) algorithm \cite{van2016deep}, (ii) \textit{policy-based} methods that learn the optimal policy function, and (iii) \textit{actor-critic} methods that combine value-based and policy-based methods.

%DRL combines neural networks with RL algorithms. Neural networks approximate functions, which are especially useful when the state space and the action space are too large to be completely known and belong to continuous space. 

In this paper, we use two actor-critic algorithms, namely, DDPG \cite{lillicrap2015continuous} and D4PG \cite{barth2018distributed}. These methods are suitable to optimize beamforming, since in both DDPG and D4PG, the action space can be continuous. Moreover, for D4PG, the exploration can be distributed among multiple agents. The DDPG algorithm uses state-action Q-value critic based on deep Q-learning \cite{van2016deep} and updates the policy using its critic gradients. D4PG builds on the DDPG approach by making several enhancements such as the Q-value estimation and the distributed collection of experiences. Note that an \textit{experience} is the process of exploring a new action by executing it in the environment or simulating it, thereby receiving some reward and observing the new state. 

%Later in this paper, we will provide detailed descriptions of the DDPG and D4PG algorithms.
%===============SubSection======================================================
\subsection{DRL Agent for Beamforming Optimization}
\label{subsec:DRLAgent}  
For a cell-free network with $M$ APs and $K$ UEs ({Fig.~\ref{System_Model}}), we develop a DRL model to optimize the beamforming matrix $\bm{W}$. This matrix  includes beamforming vectors of all UEs within the network coverage area given the complete CSI. Note that, since the beamforming vector of the $k$-th UE is given by $\bm{w}_{k}=\left[w_{1k}\dots w_{Mk}\right]$,  the matrix $\bm{W}$ has a dimension of $M \times K$. 

We cast the beamforming optimization problem as an MDP problem. We then train a DRL model where an agent learns by interacting with a cell-free network as the environment. Such a DRL model can be implemented either centrally or in a distributed manner. The design of the environment for a cell-free network includes the definition of state $\bm{s}$, the action $\bm{a}$, and the immediate reward function $\bm{r}$, needed for the DRL algorithm to estimate the policy and the Q-values. To improve the training process, new actions need to be explored by the agent. Therefore, we add a noise generated from a random process at each action taken in the training phase. Although originally an Ornstein–Uhlenbeck (OU) process was proposed for exploration, later results showed that, an uncorrelated, zero mean Gaussian noise gives the same performance [38]. Therefore, due to its simplicity of implementation, the Gaussian noise process is preferred to the OU process for exploration. The state  can be any key performance indicator. While the action of this model is the optimization variable $\bm{W}$ (the beamforming matrix) of the problem $\textbf{P}_1$, the reward can be any performance metric that jointly quantifies the performance of all active UEs. In {Table \ref{tab:ddpg_center}}, we summarize the design parameters of a DRL model and the corresponding measures in the cell-free network.  
%----------
\begin{table}[h]
\centering
\caption{Design parameters for the DRL model}
\begin{tabular}{|c|c|}
\hline
\textbf{Environment Variables} & \textbf{System Equivalence} \\
\hline \hline
State $s = \{s_1, ...., s_K\}$ & UE SINR: $\{ \gamma_1,...,\gamma_K\}$ \\
\hline
Reward $r$ & Sum-rate for all UEs: $\sum_{k=1}^K \log_2(1 + \gamma_k)$ \\
\hline
Action $a$ &  Beamforming matrix: $\bm{W}$ \\
\hline
\end{tabular}
\label{tab:ddpg_center}
\end{table}
%=============================================================Section=============================================================
\section{DRL-Based Centralized and Distributed Beamforming Methods}
\label{sec:CentralBeam}
%=====================SubSection=============================================================
\subsection{The DDPG Algorithm: A Fully Centralized Solution}
\label{subsec:ddpg-algo}
 In this section, we propose a DRL-based centralized solution for the beamforming problem in (\ref{OptimizationProb}) using the DDPG learning algorithm. This solution serves as a benchmark for the other beamforming techniques in the subsequent sections.
The actor-critic algorithm in DDPG can handle continuous state space and continuous action space. Since the elements $w_{ij}$ of $\bm{W}$ are  continuous in the range $[0,1]$, to find the optimal beamforming matrix $\bm{W} \in \left[0,1 \right]^{M \times K}$, we use DDPG. It uses two neural networks as function estimators: (i) the critic, $Q(s,a)$, whose parameters are $\theta^Q$ and calculates the expected return given state $s$ and action $a$; (ii) the actor, $\mu(s)$, whose parameter is $\theta^\mu$ and determines the policy. In DDPG, the actor directly maps states to actions instead of outputting a probability distribution across a discrete action space. The starting point for learning an estimator to $Q^*(s,a)$ is the Bellman equation given by
\begin{equation}\label{bellman_equation}
Q^*(s,a) = \mathbb{E}_{(s,a,r,s') \in \mathcal{R}} \left[ r(s,a) + \zeta \;\underset{a'}{\text{max}}\;Q^*(s',a') \right],
\end{equation}
where $\mathcal{R}$ is the set of the experiences and $\zeta$ is the discount factor.

Computing the maximum over actions in the target is quite challenging in the continuous action spaces. In particular, the training of the Q-value function involves minimizing the loss function $L(\theta^Q)$ in equation (\ref{eq:q_loss}). However, the target value $y_t$ that we want our Q-value to be close to depends on the weights of critic network $Q$. This dependency will lead to an instability in the learning of critic network $Q$ which will affect the learning of the policy $\mu$.
\begin{equation}
\begin{split}
\label{eq:q_loss}
&L(\theta^Q) = \mathbb{E}_{s_t \sim \rho^\mu, a_t \sim \beta, r_t \sim E} \Big[ (Q(s_t,a_t)|\theta^Q)-y_t)^2\Big]. \\
&y_t = r(s_t,a_t)+\zeta Q(s_{t+1},\mu(s_{t+1})|\theta^Q).
\end{split}
\end{equation}
To mitigate this instability, in \cite{lillicrap2015continuous} a lagged version of the actual critic network and the actor network are created and updated through the Polyak averaging at each iteration to enhance the convergence chances.
DDPG handles this coupling and instability by using two target networks, namely, a critic target network $Q'(s,a)$ and a policy target network $\mu'(s)$. These two networks use a set of parameters $\theta^{Q'}$ and $\theta^{\mu'}$ updated by \textit{Poylak averaging} with factor $\tau$ as shown at (\ref{Theta_1_DDPG}) and (\ref{Theta_2_DDPG}) in \textbf{Algorithm 1}. These additional networks represent a lagged version of the actual critic network and actor network that are updated at each iteration and they increase the chance of convergence of the algorithm.
Therefore, in \textbf{Algorithm \ref{algo:ddpg}}, the critic network minimizes the loss function presented at (\ref{Loss_Function_DDPG}). In policy learning, DDPG learns $\theta^\mu$ that maximizes $Q$. At every round, it maximizes the expected return as
\begin{equation}\label{policy_loss}
J(\theta^\mu) = \mathbb{E} \Big[Q(s,a)\Big|s = s_t, a=a_t\Big],
\end{equation}
and updates the weights $\theta^\mu$ by following the gradient of (\ref{policy_loss})
\begin{equation}
\label{policy_gradient}
\nabla J_{\theta^\mu}(\theta) \approx \nabla_a Q(s,a). \nabla \mu(s|\theta^\mu).
\end{equation}
%
%--------------------
\begin{algorithm}
\caption{DDPG learning process \cite{lillicrap2015continuous} (implemented in the ECP)}
\label{algo:ddpg}
\begin{algorithmic}[1]
\State Randomly initialize the actor network $\mu(s|\theta^\mu)$ and the critic network $Q(s, a|\theta^Q)$ with weights $\theta^\mu$, $\theta^Q$.
\State Initialize the target networks $Q'$ and $\mu'$ with weights $\theta^{Q'} \leftarrow \theta^{Q}$ and $\theta^{\mu'} \leftarrow \theta^\mu$.
\State Initialize the replay buffer $\mathcal{R}$.
\For{episode $ = 1\; to\; $ max-number-episodes}
\State Initialize a random process $\mathcal{N}$ for action exploration.
\State Observe the initial state $s_1$.
\For{$t = 1 \; to \;$ max-episode-steps}
\State Perform action $a_t = \mu(s_t) + \mathcal{N}_t$. Observe reward $r_t$ and the next state $s_{t+1}$.
\State Store the transition $(s_t,a_t,r_t,s_{t+1})$ in $\mathcal{R}$.
\State Let $N$ be the batch size. Sample random mini batch of $N$ transitions from $\mathcal{R}$.
\State Update the critic by minimizing the loss
\begin{equation}\label{Loss_Function_DDPG}
L = \frac{1}{N} \sum_{j}(y_j - Q(s_j,a_j|\theta^Q))^2,
\end{equation}
\begin{equation}
y_j = r_j + \zeta Q'(s_{j+1},\mu'(s_{j+1})|\theta^{\mu'})|\theta^{Q'}).
\end{equation}
\State Update the actor policy using sampled policy gradient ascent as
\begin{equation}
\nabla_{\theta^\mu}J \approx \frac{1}{N} \sum_j \nabla_a Q(s,a|\theta^Q)|_{s=s_j,a=\mu(s_j)} \nabla_{\theta^\mu} \mu(s|\theta^\mu)|_{s=s_j}.
\end{equation}
\State Update the target networks as
\begin{equation}\label{Theta_1_DDPG}
\theta^{Q'} \leftarrow \tau \theta^Q + (1-\tau)\theta^{Q'},
\end{equation}
\begin{equation}\label{Theta_2_DDPG}
\theta^{\mu'} \leftarrow \tau \theta^\mu + (1-\tau) \theta^{\mu'}.
\end{equation}
\EndFor
\EndFor
\end{algorithmic}
\end{algorithm}
This update rule represents the Deterministic Policy Gradient theorem \cite{pmlr_v32_silver14}. The term $\nabla_a \,Q(s,a)$ is obtained from the backpropagation of the Q-network $Q(s,a|\theta^Q)$ w.r.t the action input $\mu(s|\theta^{\mu})$. \textbf{Algorithm \ref{algo:ddpg}} summarizes the DDPG learning process with settings presented in Table \ref{tab:ddpg_center} (i.e. the states $s_t$, actions $a_t$ and rewards $r_t$ are defined as UE SINR $\{ \gamma_1,...,\gamma_K\}$,  the beamforming matrix $W$, and the sum-rate $\sum_{k=1}^K \log_2(1 + \gamma_k)$, respectively, at the $t$-th learning step).
%=====================SubSection=============================================================
\subsection{D4PG-Based Beamforming With Distributed Agents' Experience}
\label{sec:D4PGBeam}
In the previous section, we discussed the fully centralized DDPG-based beamforming scheme as a replacement of conventional centralized optimization at the ECP. In this section, we take a step toward distributed beamforming in cell-free networks. We propose a DRL-based beamforming scheme which exploits the experiences of the  distributed actors (e.g. eAPs) that belong to a centrally located agent (e.g. ECP). This scheme is based on the distributed distributional deep deterministic policy gradient (D4PG)  algorithm \cite{d4pg}. 
Generally, the D4PG algorithm applies a set of improvements on DDPG and make it run in a distributional fashion. These improvements enable D4PG to outperform  DDPG \cite{DDPG_Vs_D4PG}. This is achieved by having a multiple actor neural networks to gather independent experiences and feed it to a single replay buffer. However, it contains a single critic network that samples the independently gathered experiences from the replay buffer to explore a new Q-value. Traditionally,  the independent actors, the replay buffer, and the critic of the D4PG are implemented within the same operation area/module. In this paper, we propose an implementation for the D4PG by setting the number of independent actors to $M$ (equal to the number of eAPs) and placing an actor at each eAP, and place the replay buffer and the critic network at the ECP. As illustrated in Fig. \ref{fig:d4pg_model}, the actor networks implemented in the eAPs generate independent experiences that are sent to the ECP to be pushed into the replay buffer, which combines all experiences explored by the actor networks. In \textbf{Algorithm \ref{algo:actor}}, we describe the exploration technique used in the actor network. This approach improves the quality of training data since the eAPs simultaneously generate multiple experiences with different exploration processes. 

%-----------
\begin{figure*}[htb]
\centering
\includegraphics[scale = 0.37]{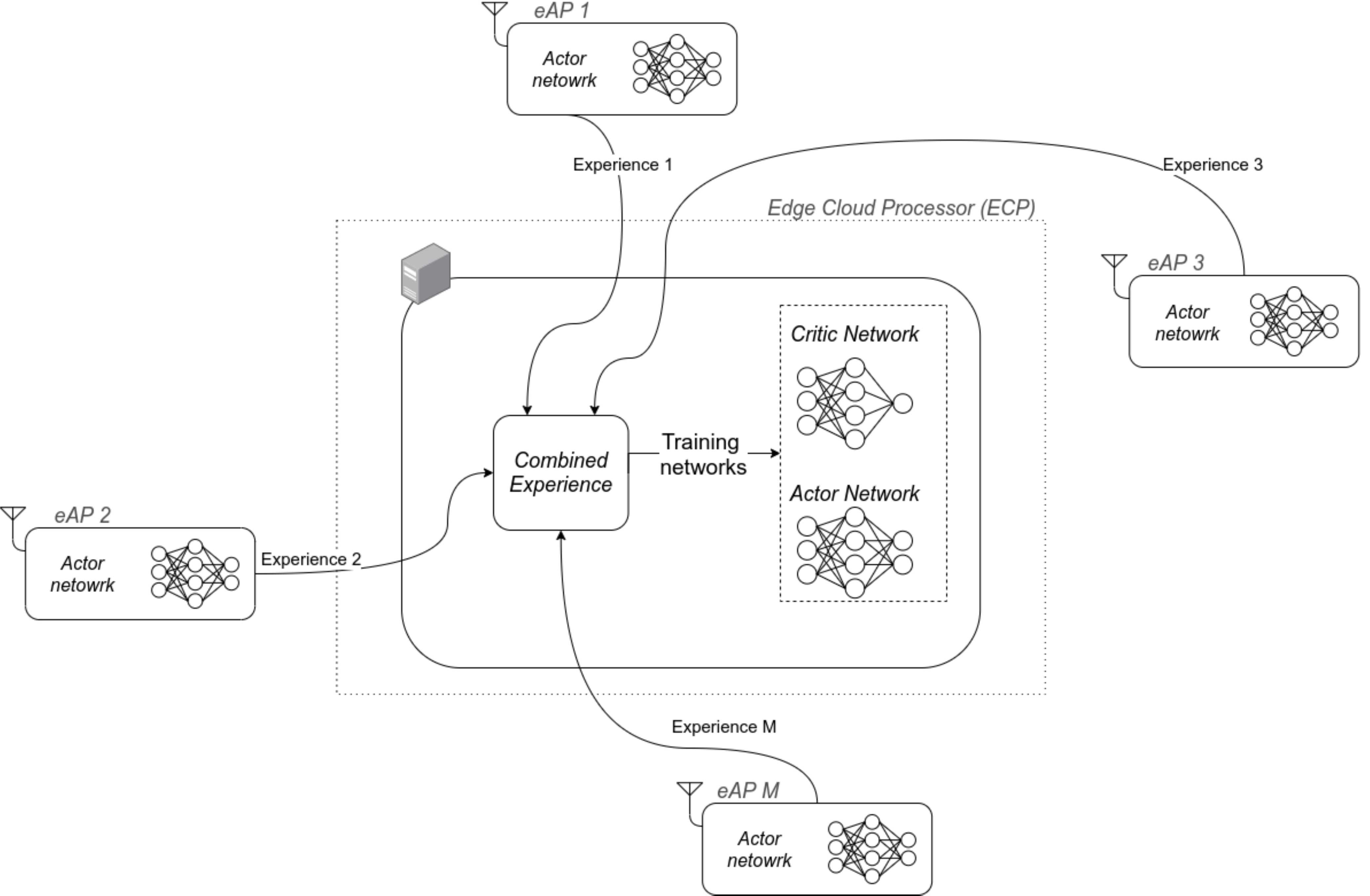}
\caption{D4PG-based beamforming design.}
\label{fig:d4pg_model}
\end{figure*}
%-----------
The design parameters of the DRL model for D4PG are represented in Table \ref{tab:ddpg_center}, therefore the states $s_t$, actions $a_t$ and rewards $r_t$ are defined as the UE SINR $\{ \gamma_1,...,\gamma_K\}$, the beamforming matrix $W$, and the sum-rate $\sum_{k=1}^K \log_2(1 + \gamma_k)$, respectively, at the $t$-th learning step. D4PG is an actor-critic method that enhances the DDPG algorithm to perform in a distributed manner, thereby improving the estimations of the Q-values. Below we discuss these enhancements to the D4PG algorithm, which we explicitly describe in \textbf{Algorithm \ref{algo:d4pg}}.
%----------subsubsection-----------------------------------------------------------
\begin{algorithm}
%\footnotesize
\caption{D4PG-Based Beamforming  With Distributed Agents' Experience}
\label{algo:d4pg}
\begin{algorithmic}[1]
\State \textbf{Input:} batch size $P$, trajectory length $N$, number of actors $K$, replay size $R$, initial learning rates $\alpha_0$ and $\beta_0$, time for updating the target networks $t_{target}$, time for updating the actors policy $t_{actor}$.
\State Initialize actor and critic networks $(\theta^\mu$,$\theta^Z)$ at random.
\State Initialize target weights $(\theta^{\mu'}$,$\theta^{Z'}) \leftarrow (\theta^\mu$,$\theta^Z)$.
\State Launch $K$ actors and replicate their weights $(\theta^{\mu'}$,$\theta^{Z'})$.
\For{$t=1,\dots,T$}
\State Sample $P$ transitions $(s_{i:i+N},a_{i:i+N},r_{i:i+N})$ of length $N$ from replay with priority $p_i$.
\State Construct the target distributions $Y_i = \sum_{n=0}^{N-1} \zeta^n r_{i+n} + \zeta^N Z'(s_{i+N},\mu'(s_{i+N}|\theta^{\mu'})|\theta^{Z'})$.
\State Compute the actor and critic updates
$$ \delta_{\theta^Z} = \frac{1}{P} \sum_{i} \nabla_{\theta^Z} (R\,p_i)^{-1} d(Y_i,Z(s_i,a_i|\theta^Z)),$$
$$\delta_{\theta^\mu} = \frac{1}{P} \sum_{i} \nabla_{\theta^\mu} \mu(s_i|\theta^\mu) \mathbb{E}[\nabla_a Z(s_i,a|\theta^Z)]|_{a=\mu(s_i|\theta^\mu)}.$$
\State Update network parameters $\theta^\mu \leftarrow \theta^\mu + \alpha_t\; \delta_{\theta^\mu}$,\; $\theta^Z \leftarrow \theta^Z + \beta_t \;\delta_{\theta^Z}$.
\State If $t \,\,\mathbf{mod}\,\, t_{target} = 0$, update the target networks $(\theta^{\mu'},\theta^{Z'}) \leftarrow (\theta^\mu,\theta^Z)$.
\State If $t \,\,\mathbf{mod}\,\, t_{actors} = 0$, replicate network weights to the actors.
\EndFor
\State \textbf{Return} policy parameters $\theta^\mu$.
\end{algorithmic}
\end{algorithm}
%----------------
\begin{algorithm}
%\footnotesize
\caption{Actor}
\label{algo:actor}
\begin{algorithmic}[1]
\Repeat
\State Set $t=0$
\State Initialize random process $\mathcal{N}$ at each actor
\State At step t, select action $a_t = \mu(s_t|\theta^\mu) + \mathcal{N}_t$. Receive reward $r_t$ and observe state $s_t'$.
\State Send the experience $(s_t,a_t,r_t,s_t')$ to ECP to be stored in the replay buffer.
\Until{learner finishes, i.e. \textbf{Algorithm \ref{algo:d4pg}} stops.}
\end{algorithmic}
\end{algorithm}
   
%===============SubSection======================================================
The main features of the D4PG algorithm are described below.

\subsubsection{Distributional Critic}
\label{subsubsub:Dist}
In D4PG, the Q-value is a random variable following some distribution $Z$ with parameters $\theta^Z$, thus $Q_{\theta^Z}(s,a) = \mathbb{E}[Z(s,a|\theta^Z)]$. The objective function for learning the distribution minimizes some measure of the distance between the distribution estimated of the target critic network and that of the critic network, e.g. the binary cross-entropy loss. Formally,
\begin{equation}
L(\theta^Z) = \mathbb{E}[d(T_{\mu}\;Z'(s,a|\theta^{Z'}),Z(s,a|\theta^Z)],
\end{equation}
where $T_\mu$ is the Bellman operator. The deterministic policy gradient update yields
\begin{equation}
\begin{split}
\nabla_{\theta^\mu} J(\theta^\mu) &\approx \mathbb{E}_\rho[\nabla_{\theta^\mu} \mu(s|\theta^\mu) \nabla_a Q(s,a|\theta^Z)|_{a=\mu(s|\theta^\mu)}] \\
& = \mathbb{E}_\rho[\nabla_{\theta^\mu} \mu(s|\theta^\mu) \mathbb{E}[\nabla_a Z(s,a|\theta^Z)]|_{a=\mu(s|\theta^\mu)}].
\end{split}
\end{equation}
%===============SubSection======================================================
\subsubsection{N-step Returns}
\label{subsubsub:Step}
An agent in D4PG computes the $N$-step Temporal Difference (TD) target instead. Formally,
\begin{equation}
(T_{\mu}^N Q)(s_0, a_0) =  r(s_0, a_0) + 
 \mathbb{E}\left[\sum_{n=1}^{N-1} r(s_n,a_n) + \zeta^N Q(s_N,\mu_\theta(s_N))|s_0,a_0\right].
\end{equation}
%===============SubSection======================================================
\subsubsection{Multiple Distributed Parallel Actors}
\label{subsubsub:Mult}
This process takes place in parallel within $K$ actors, each one generating samples independently. The samples are collected in a replay buffer from which a learner samples batches to update the weights of the networks.
%===============SubSection======================================================
\subsubsection{Prioritized Experience Replay (PER)}
\label{subsubsub:Peri}
Finally, D4PG collects $R$ samples from the replay buffer with non-uniform probability $p_i$. The $i$-th sample is selected with priority $(Rp_i)^{-1}$ that also indicates the importance of the sample.

%===========================================Section====================================================== 
\subsection{DRL-Based Beamforming With Distributed Learning}
\label{sec:DistLearn}
%===============SubSection======================================================
%\subsection{Distributed Learning Algorithm}
In the previous section, we proposed a D4PG learning method to train a policy that predicts the beamforming matrix given a cell-free network environment. We enhanced the learning performance by allocating an agent per eAP and then utilizing the distributed experience collected from several agents. The enhanced method converges faster and shows better performance compared to the fully centralized DDPG solution. However, the proposed D4PG method still conducts the learning process at the ECP that involves a large body of the computational tasks. 

In this section, we propose a cell-free beamforming scheme, which, in addition to distributing the agents' experiences, splits the learning process of the DRL among all eAPs. In such a model, the eAPs divide the computational tasks equally and the ECP only performs limited control and coordination task. In this scheme, every eAP is responsible to find the optimal beamforming vector for all UEs. To this end, all vectors of the overall  beamforming matrix are considered to be constants, to be simultaneously found by other eAPs. As an example, for a cell-free network with $M$ eAPs and $K$ UEs, eAP$_m$ is responsible to optimize $\bm{\omega}_m=\bm{W}\left(m,[1,~ \dots,~k] \right)$.
Thus, eAP$_m$ solves the following subproblem:
\begin{equation}
\begin{aligned}
\textbf{P}_m:&~ \underset{ \bm{\omega}_m\in \left[0~1 \right]^{1 \times K}}{\text{maximize}}~
%& \text{\hspace{-55mm}} \sum_{k=1}^K\log_2\left(1+\gamma_{k}  \right)\\
\sum_{k=1}^K\log_2\left(1+\gamma_{k}  \right)\\
\text{S.t.}~&~\textbf{C}_1: \sum_{m=1}^{M} \left(w_{m \delta_l}^2-\sum_{i=\delta_l+1}^{l}w_{mi}^2 \right)\bar{\gamma}_{{m}l}\geq \bar{P_s}, \\
&  ~  \textbf{C}_2:w_{mk}^2\in \left[0,1\right],
\\
&~~\forall k,  \forall~\delta_l= 1, \dots, l-1\;\mbox{, and}\; l=2, \dots, K.\\
\end{aligned}\label{OptimizationProb_m}
\end{equation}

In the following, we present a new system design for the cell-free network environment where the eAPs interact with each other to find the optimal beamforming vector. We intend to solve the optimization problem in (\ref{OptimizationProb}) through multiple eAPs solving the problem in (\ref{OptimizationProb_m}). Therefore, as shown in Fig. \ref{DRL_Dist_Learning}, we distribute the learning process among the eAPs by implementing a DDPG agent in each eAP with its local experience buffer, actor network, and critic network.

%--------------
\begin{figure*}[htb]
\centering
\includegraphics[width=\textwidth]{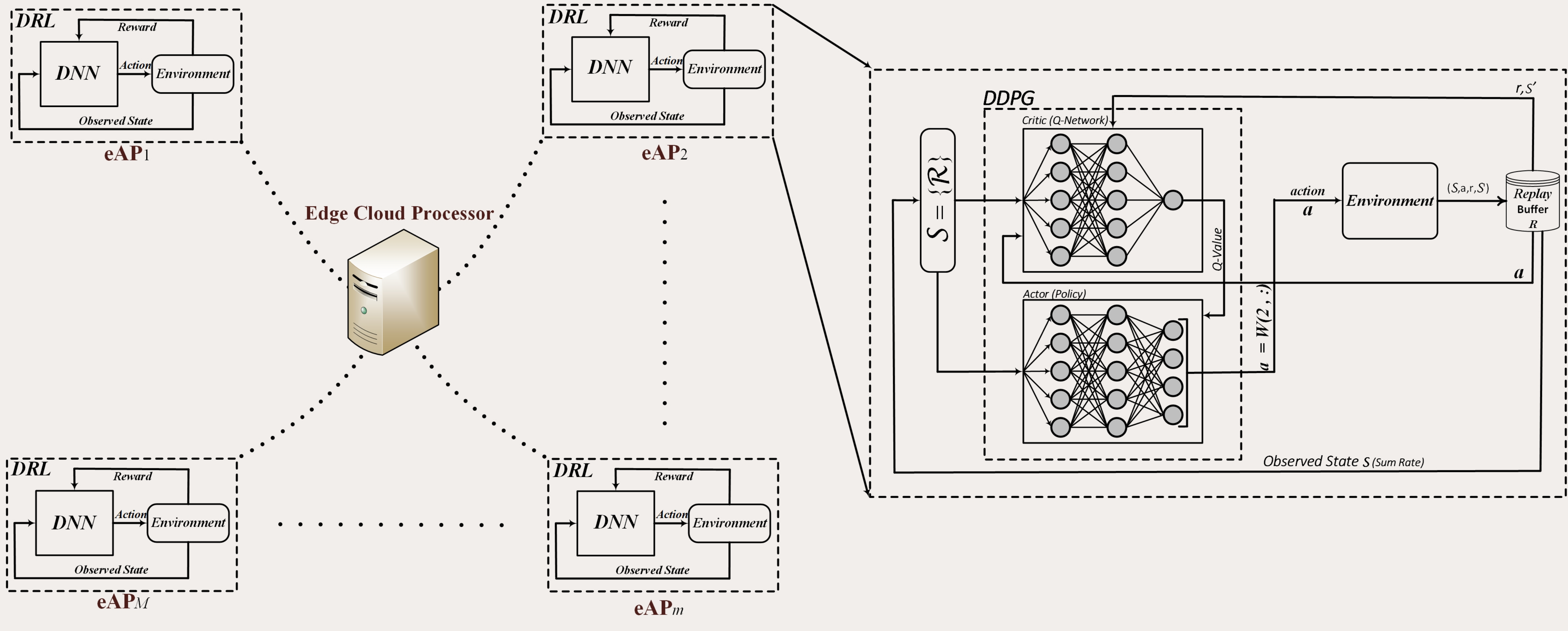}
\caption{Example of the proposed distributed learning-based DRL beamforming.}\label{DRL_Dist_Learning}
\end{figure*}
%-----------

The design parameters of the DRL model are slightly changed compared to Table \ref{tab:ddpg_center}. The states $s_t$ and $r_t$ are the UE SINR $\{ \gamma_1,...,\gamma_K\}$ and sum-rate $\sum_{k=1}^K\log_2\left(1+\gamma_{k}  \right)$ (as before), but the actions $a_t$ are changed to $\bm{\omega}_m$, since the $m$-th eAP learns to optimize the $m$-th row of the beamforming matrix $\mathbf{W}$. Beamforming is then optimized as follows: The first step is the DDPG process that includes: (i) generating experiences, and (ii) updating the network parameters using Bellman equation (\ref{bellman_equation}), and the policy gradient in (\ref{policy_gradient}). Each agent learns to optimize one row of the beamforming matrix. Here, the ECP, which is connected to all the eAPs through backhaul links, acts as the  \textit{coordinator} to facilitate sharing of the optimized rows with other agents. At each episode, the eAPs discover multiple actions. They select the action with the highest reward in terms of the maximum normalized sum-rate during the episode and send it to the coordinator. The coordinator receives all the rows from the eAPs, and then concatenates them to create a new beamforming matrix. Finally, it broadcasts the new matrix to all eAPs. The number of updates of the beamforming matrix, i.e. the horizon, can be selected based on the required accuracy. \textbf{Algorithm \ref{algo:ddpg_dist}} summarizes the learning process in each eAP, and \textbf{Algorithm \ref{algo:coor_dist}} summarizes the coordination process among the eAPs by the the ECP.
%------------------

\begin{algorithm}
%\footnotesize
\caption{Local DDPG (implemented in each eAP)}
\label{algo:ddpg_dist}
\begin{algorithmic}[1]
\State Initiate the cell-free network environment with random beamforming matrix $\bm{W}_{\textrm{local}}$.
\State Initialize actor network $\mu(\cdot)$ and critic network $Q(\cdot)$ with parameters $\theta^\mu$ and $\theta^Q$, respectively.
\State Initiate the random process $\mathcal{N}$.
\State Initiate $r_{\textrm{max}} = 0$.
\For{each episode}
\State Observe initial state $s$ from the cell-free network environment.
\For{each step}
%\State Do an experience.
\State Apply $a = \mu(s) + \mathcal{N}_t$ on the cell-free network environment and observe new state $s'$ and reward $r$. 
\State Store $(s,a,r,s')$ in the replay buffer $\mathcal{R}$.
\State Update $\theta^\mu$ and $\theta^Q$ weights with a batch from $\mathcal{R}$ using Bellman equation (\ref{bellman_equation}) and policy gradient (\ref{policy_gradient}).
\If{$r \geq r_{\textrm{max}}$}
\State $a_{\textrm{optimal}} = a$.
\State $r_{\textrm{max}} = r$.
\EndIf
\EndFor
\If{update time of beamforming matrix}
    \State Send $a_{\textrm{optimal}}$ to the coordinator.
    \State Receive $\bm{W}_{\textrm{local}}$ from the coordinator.
    \State $\bm{W}_{\textrm{local}} = \bm{W}_{\textrm{new}}$.
\EndIf
\EndFor
\end{algorithmic}
\end{algorithm}
\begin{algorithm}
%\footnotesize
\begin{algorithmic}[1]
\caption{Coordinator (implemented in the ECP)}
\label{algo:coor_dist}
\Repeat
\State $\bm{W}_{\textrm{new}}=0$
\Repeat
\If{$a_{\textrm{optimal}}$ is received from $i$-th eAP}
\State $\bm{W}_{i, \,\textrm{new}}= a_{\textrm{optimal}}$.
\EndIf
\Until{$a_{\textrm{optimal}}$ is received from every eAP.}
\State Broadcast $\bm{W}_{\textrm{new}}$ to all agents.
\Until{all agents stop training.}
\end{algorithmic}
\end{algorithm}

\section{Complexity Analysis and Communication Overhead}
\label{sec:Complexity}
%================================================================================
As we have mentioned before, the beamforming optimization problem (\ref{OptimizationProb}) is non-convex. Conventional approaches to solve a non-convex problem include iterative algorithms such as exhaustive search, steepest descent, gradient ascent, and interior-point methods\footnote{Note that the objective function in (\ref{OptimizationProb}) is twice differentiable.}. Considering an exhaustive search, we quantize the beamforming vector of $k$-UE, i.e. $\bm{\omega}_k=\left[w_{1k} \dots w_{Mk} \right]$, by a certain step size $\Delta$. For the beamforming optimization problem, this results in a complexity of $O\left(\left(\frac{1}{\Delta}\right)^{M+K}\right)$. The combinatorial complexity renders the conventional beamforming techniques impractical, in particular, for dense cell-free networks. This calls for novel methods such as DRL-based solutions that retain low complexity while guaranteeing efficiency.

To evaluate the time-complexity of a DNN, the conventional measure is the \textit{floating-point operations per second} (FLOPs). For any fully connected layer $L_i$ of input size $I_i$ and output size $O_i$, the number of FLOPs is given by
\begin{equation}
\text{FLOPS}(L_i) = 2 \,I_i\,O_i.
\end{equation}
Thus the total number of FLOPS of the DRL-based method during the inference for a policy with $L$ hidden layers yields:
\begin{equation}
\begin{aligned}
\text{FLOPs}_{\text{DRL}}&=\sum_{i=1}^{3}\, \text{FLOPs}(L_i) \\
&= 2 \cdot \Big( H_1 \cdot\mathcal{|S|} + H_L \cdot \mathcal{|A|} + \sum_{i=2}^L H_{i-1} \times H_i \Big)\\
&=2 \cdot \Big( H_1 \times K + H_L \left(K\times M\right) + \sum_{i=2}^L H_{i-1} \times H_i \Big),    
\end{aligned}
\end{equation}
where $H_1$, $H_L$, and $H_i$ denote the size of the first hidden layer, $L$-th hidden layer, and $i$-th hidden layer, respectively.

Table \ref{Comparision} compares the order of complexity of inference, as well as the convergence rate, of the three proposed DRL-based methods. Note that the policy implemented in this paper has 2 hidden layers of sizes 258 and 128.
%-------------- 
\begin{table}[h!]\scriptsize
\centering
\caption{Complexity analysis}
\begin{tabular}{|c||c|c|}
\hline 
\textbf{Model} & \textbf{FLOPS for inference} & \textbf{Convergence} \\
\hline
\hline
\shortstack{MMSE Solution\\ \textcolor{white}{.}} 
 & \shortstack{$\textit{O}\left(\left[M\times K\right]^2\right)$\\ \textcolor{white}{ }}
& \shortstack{Polynomial\\ Convergence} \\
\hline
Centralized (DDPG) & $2\cdot\Big(256 \times 128$ + 256$\cdot{K}$ + 128$\cdot{(M\times K)\Big)}$  & Quadratic \\
\hline
\shortstack{Distributed Experience\\ (D4PG)}  & \shortstack{$2\cdot\Big(256 \times 128$ + 256$\cdot K$ + 128$\cdot(M\times K)\Big)$\\ \textcolor{white}{.}} & \shortstack{Quadratic\\ \textcolor{white}{.}} \\
\hline
Distributed Learning  & $2\cdot\Big(256 \times 128$ + 256$\cdot K$ + 128$\cdot K\Big)$ & Quadratic \\
\hline
\end{tabular}
\label{Comparision}
\end{table}
%---------------
The number of FLOPS during the inference is mainly determined by the matrix multiplications of the policy network, which has four layers with size $\mathcal{|S|}$, 256, 128, and $\mathcal{|A|}$. The number of FLOPs for a connected layer is twice the product of the size of the hidden layers. That is, given a hidden layer of size 128 that has a previous layer of size 256, the weight matrix should be of size $128 \times 256$ and the FLOPs is $2 \times ( 256 \times 128)$. For D4PG, which uses distributed training, the number of FLOPS is equal to that of DDPG since the feed-forward operation on the policy network remains identical during the inference. In contrast, for the fully distributed approach, each of the $M$ eAPs predicts one row of the beamforming matrix. Therefore, the number of FLOPS is smaller than that of D4PG with distributed experience and DDPG. In the next section, we will investigate the performance of DRL models in terms of transmission sum-rate and the convergence time.

Also, note that, for the D4PG method, efficient communication required for sharing the experiences, generated locally at the eAPs, with the ECP can be challenging. Moreover, in case of heterogeneous processing capabilities at the eAPs, a synchronization method will be required at the ECP for reception of the experiences from the eAPs. This method should avoid any bias toward experiences coming from the actors that  contribute more to the replay buffer.
%===================================================================================
%\textbf{Communication Overhead:} 
In a cell-free network, optimizing the uplink beamforming necessitates some level of control signaling among the distributed entities (eAPs and UEs). In this regard, each of our proposed beamforming solutions imposes a different requirement. Considering CSI, all of the solutions require the same amount of signaling overhead, which is due to the collection of the received pilot signals at each eAP. This is followed by application of the MMSE algorithm and estimation of the CSI among all UEs and the intended eAP. To find the beamforming matrices, the fully distributed solution requires the highest amount of signaling exchanges to transfer the beamforming vectors among different eAPs. This is obviously the cost paid to enable a distributed learning-based solution.  

%===========================================Section======================================================
\section{Numerical Analysis}
\label{sec:Numerical}
%===============SubSection======================================================
\subsection{Simulation Parameters}
%
%Run the following scenarios for the three networks (centralized, distributed with DDPG, distributed with D3PG and distributed with D4PG)
%1) M=15, K=5 and $P= 10^{\frac{37-30}{10}}=5.01187233627
%$\\
%2) M=50, K=15 and $P= 10^{\frac{37-30}{10}}=5.01187233627
%$\\
%3) M=300, K=70 and $P= 10^{\frac{37-30}{10}}=5.01187233627
%$ (changed to M=70 and K=20 due to lack of resources in the approach where each eAP predicts one row of the matrix.) \\
We numerically evaluate the proposed beamforming methods in terms of transmission performance (sum-rate) and convergence rate. \textbf{Table \ref{Simulation_Parameters}} summarizes the most important parameters of the simulation setting.
%---------------
\begin{table}[h!]
\centering
 \caption{Simulation parameters}
\begin{tabular}{|c|c|}
\hline 
\textbf{Parameter} & \textbf{Value}  \\
\hline \hline
AWGN PSD per UE &
$-169$ dBm/Hz  \\
\hline
Path-loss exponent, $\kappa$ & $2$\\
\hline
Nakagami fading parameters, $\left(\mathcal{M}, \Omega\right)$  & $\left(1, 1 \right)$\\
\hline
Training sequence length, $\tau_p$ & $K$ Samples\\
\hline
Pilot transmission power, $\rho_k$ & $100$ mW, ${\forall~k}$\\
\hline
SIC sensitivity, $P_s$  & $1$ dBm\\
\hline
\end{tabular}
\label{Simulation_Parameters}
\end{table}
%---------------

For simplicity, we assume the following: (i) For every AP $m$ and every UE $k$, $h_{mk}$ is a random variable with $\mathcal{M}_{mk}=\mathcal{M}=1$ and $\Omega_{mk} = \Omega = 1$; (ii) For every AP $m$, the AWGN has PSD $\sigma_{m}/2=-169$ dBm/Hz; (iii) Concerning large-scale fading, all eAPs and UEs are uniformly distributed over a disc of radius $r=18$ meters, implying a coverage area of $1017.88 \text{m}^2$. The ECP knows the large-scale fading of each UE. Furthermore, to evaluate different DRL models, we assume that each iteration of the DRL corresponds to one coherent block (one CSI realization). Note that the DRL models can also be trained under varying CSI (at each channel use); however, this requires a larger number of iterations. We avoid such a training procedure to simplify the simulation.

%\st{To benchmark the proposed DRL schemes, we simulate the proposed centralized system model using the `gradient ascent' algorithm, which can find the local minimum of any first-order differentiable function \cite{Boyd2004}. The solution obtained by this algorithm is identical to that obtained by solving problem $\textbf{P}_1$ using the zero-forcing beamforming technique \cite{Cell_Less_3}. This also can provide a benchmark for sub-optimal distributed conjugate beamforming technique proposed in \cite{Cell_Less_3}.}
We train the proposed models by using Python and TensorFlow 2.1.0 for $1 0$ episodes with $1000$ steps per one episode. The actor- and critic networks have fully connected layers with two hidden layers of 256 and 128 neurons followed by the  \textit{Relu} activation function in each. The dimension of the final layer for the actor network and its corresponding target network is defined in the cell-free network environment, depending on the approach followed. In the first two approaches, the dimension of the output layer is equal to the number of elements in the beamforming matrix. In the third approach, it is equal to the number of elements of one row.

The constraint C1 in the problem formulation P1 is implemented by affecting a negative reward to the actions. The constraint C2 is imposed by defining a softmax activation function in the output layer which takes as input a vector of real values and normalizes it into a vector with values between 0 and 1. The hyperparameters of the DRL-model are as follows: discount factor $\zeta = 0.99$, learning rate $\nu = 0.001$ for both actor and critic networks, a Poylak averaging parameter $\tau = 0.005$, and size of experience replay buffer $|\mathcal{R}| = 10^6$. We use {\rm Adam} for the critic and actor optimizer. In the D4PG approach, we use $N$-step returns $N=5$ and $51$ atoms in the distributional representation with $V_{min}=-20$ and $V_{max}=100$ as defined in \cite{bellemare2017distributional}. We use the binary cross-entropy as the metric of the distance between distributions.

To verify the effectiveness of the proposed scheme, we evaluate the performance in terms of instantaneous reward defined by (\ref{OptimizationProb}). We first train the proposed methods under three possible network models, namely, (i) small-scale cell-free network with $M=15$ and $K=5$; (ii) medium-scale cell-free network with $M=50$, $K=15$; (iii) large-scale cell-free network with $M=70$ and $K=20$. 
%===============SubSection======================================================
\subsection{Results}
{Fig. \ref{Smale_Scale}} shows the normalized transmission sum-rate versus the overall number of training steps in the small-scale setting. {In this figure, we compare the performances of the proposed beamforming schemes with those for the two most common techniques used in the literature namely, MMSE and conjugate beamforming. Note that conjugate beamforming is considered as the simplest beamforming scheme, where by knowing the CSI, each AP multiplies the received signal by the intended UE's CSI vector without considering interference from other UEs. It can be noticed from this figure that the MMSE method gives the best performance while the conjugate beamforming method results in the worst performance.} Among the proposed DRL-based beamforming schemes,  the scheme with  distributed learning and centralized training (i.e. D4PG) exhibits the best performance. This is due to the fact that, the D4PG actors distributed among different eAPs generate more independent experiences all of which can be exploited by the critic network at the ECP.
Moreover, the performance of the fully distributed DRL-based beamforming is better than that of the fully centralized scheme (DDPG). Indeed, in a small network, since  the beamforming vector has low dimensions, the performance of the DRL method at each eAP is almost identical to those of the conventional optimization methods (such as the steepest ascent-based iterative algorithms). Furthermore, the D4PG achieves the closest performance to the MMSE beamforming solution (without learning) with about $80\%$ of the performance of the MMSE method after 6000 learning steps. {Note that the performance degradation in the DRL-based systems is the cost we pay while we achieve a significant reduction in computational complexity compared to that of  centralized MMSE method. Furthermore, even though the conjugate beamforming method can be implemented in a fully distributed manner with low processing requirements, it suffers from significant performance degradation compared to all other methods (as shown in Fig. \ref{Smale_Scale}). }

%-------------------

\begin{figure}[htb]
\minipage{0.47\textwidth}
  \includegraphics[width=\textwidth]{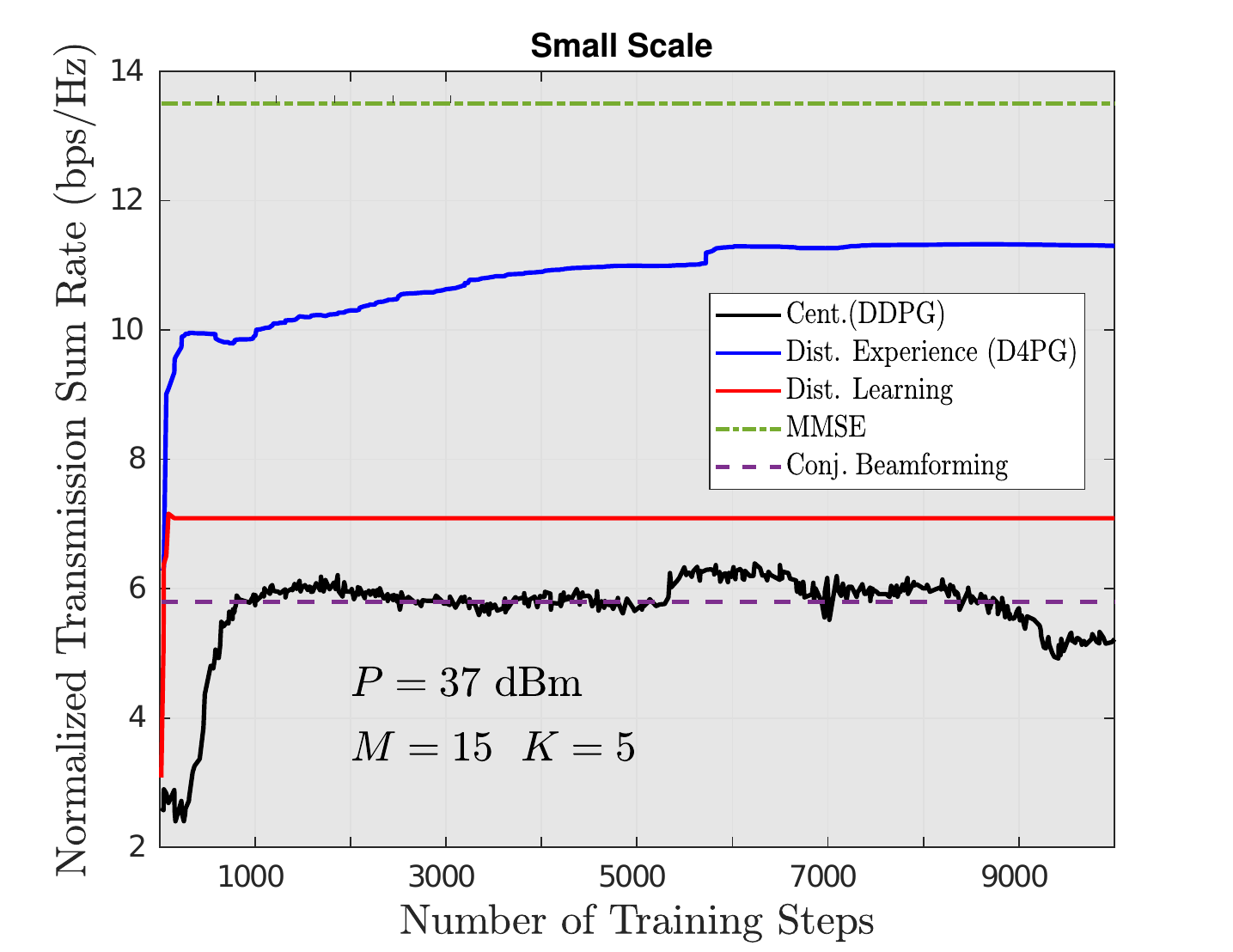}
  \caption{Performance of different models under small-scale scenario.} 
  \label{Smale_Scale}
\endminipage\hfill
\minipage{0.5\textwidth}
\includegraphics[width=\textwidth]{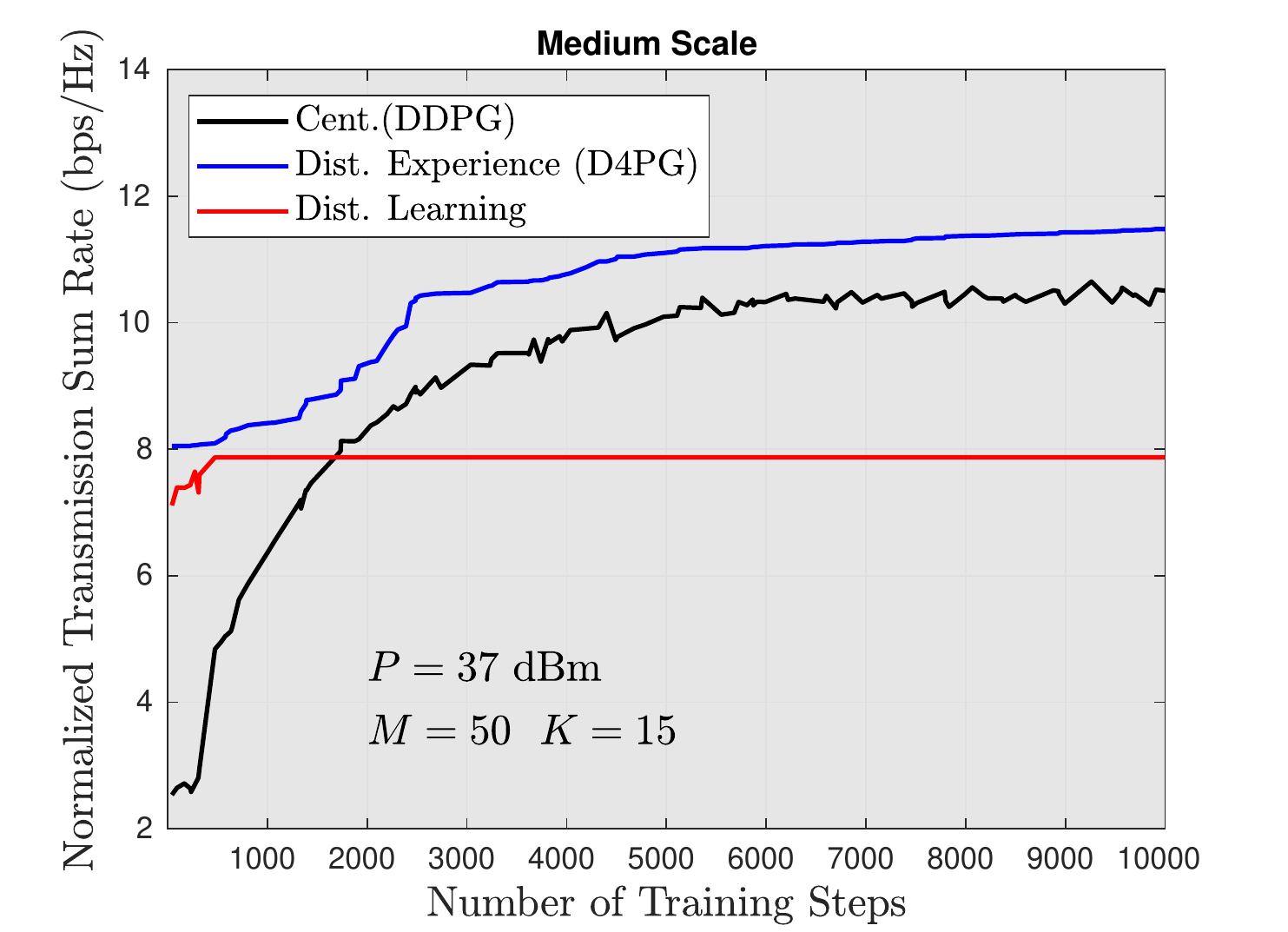}
\caption{Performance of different models under medium-scale scenario.}
\label{Medium_Scale}
\endminipage\hfill
\end{figure}

%-------------------

{Fig. \ref{Medium_Scale}} shows the results for a medium-scale cell-free network. The centralized DRL-based beamforming with distributed experience (i.e. D4PG) retains its superiority over other methods, although the performance gap with the fully centralized DDPG is smaller in this case. Moreover, the fully distributed DRL-based beamforming is no longer superior to the centralized DDPG. The reason is that in the distributed setting, every eAP uses the beamforming vectors found by the other eAPs in previous iterations. This introduces inaccuracy in treating the inter-user interference from other eAPs.

In {Fig. \ref{Large_Scale}}, we consider a larger-scale network setting. The gap between the distributed DRL-based beamforming and other methods with centralized learning (DDPG and D4PG) increases significantly. Additionally, the centralized DRL-based beamforming with distributed experience (D4PG) maintains better performance compared to the fully centralized DDPG method in the sense that it converges in fewer steps. Nevertheless, after a relatively large number of training steps, the DDPG algorithm is observed to perform slightly better than the D4PG algorithm.

%------------
\begin{figure}[htb]
\begin{center}
\minipage{0.47\textwidth}
  \includegraphics[width=\textwidth]{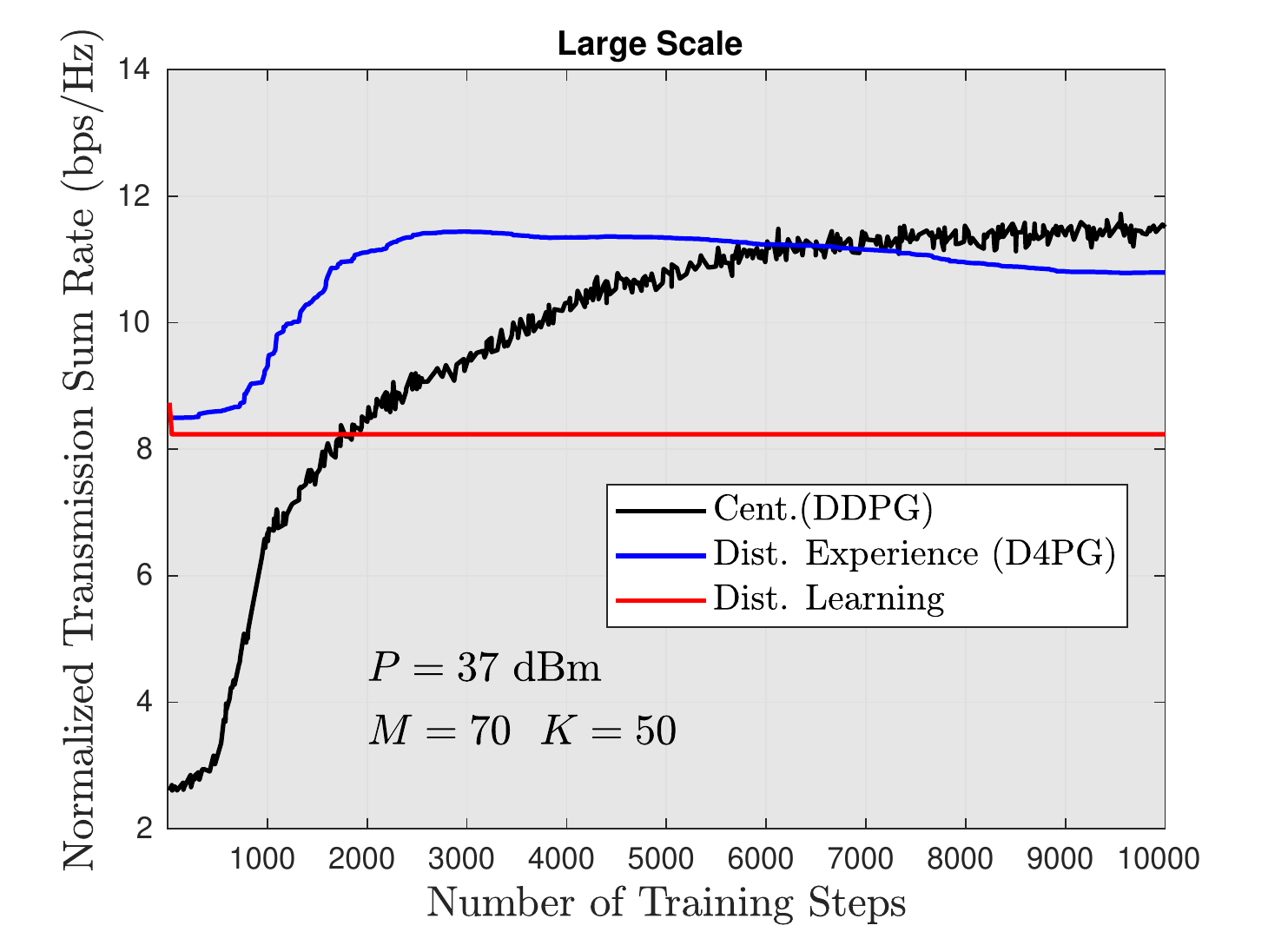}
  \caption{Performance of different models under large-scale scenario.}
  \label{Large_Scale}
\endminipage\hfill
\centering
\minipage{0.47\textwidth}
  \includegraphics[width=\textwidth]{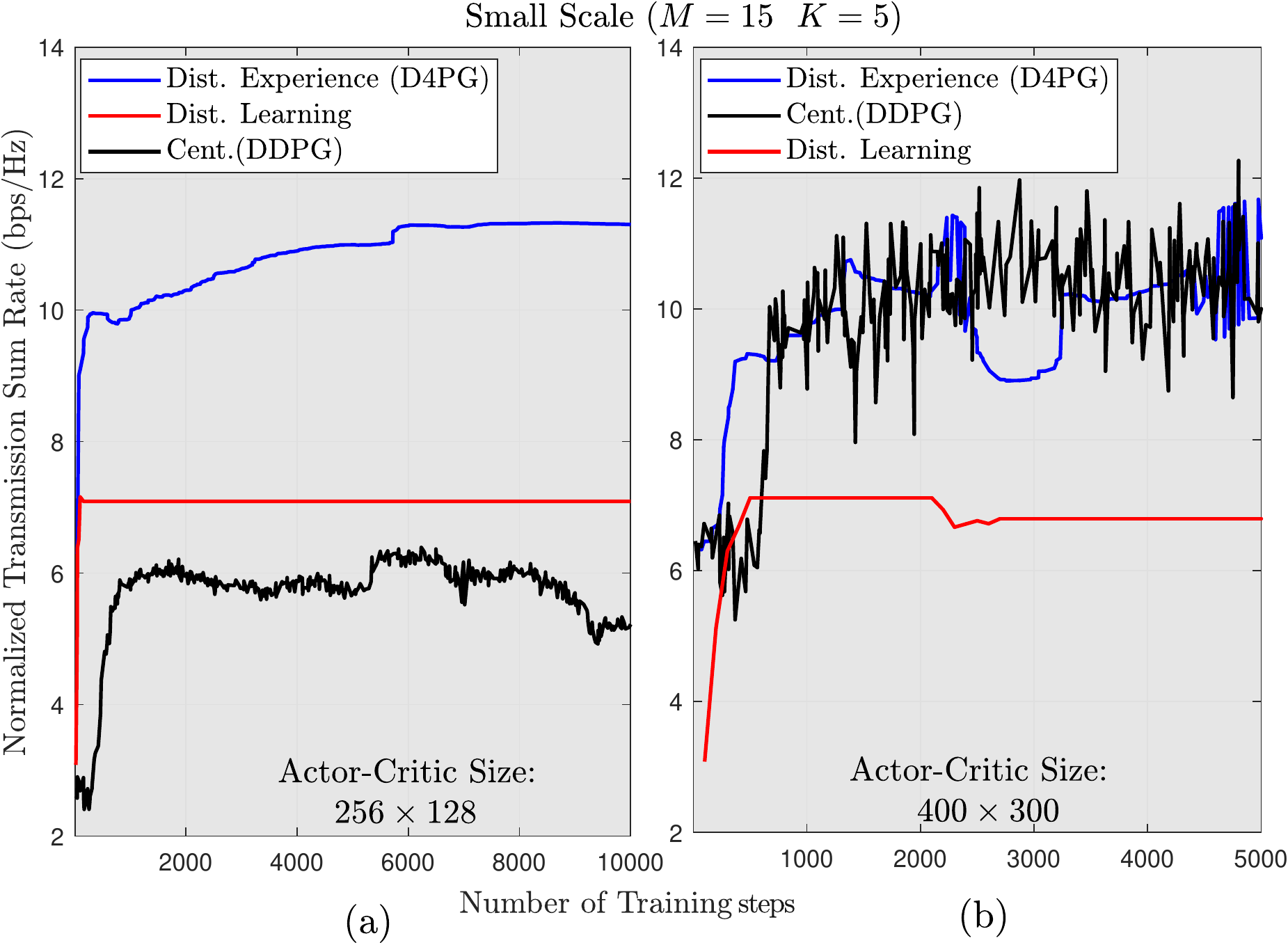}
\caption{Performance of DRL models under different actor-critic network sizes.}
\label{NN_Size}
\endminipage\hfill
\end{center}
\end{figure}

%-------------

In {Fig. \ref{NN_Size}}, we evaluate the performance of the proposed DRL models with larger sizes of actor-critic networks. Specifically, we increase the size of both the actor and critic networks from $256\times128$ to $400\times 300$. With larger sizes of the neural networks, we observe   a smaller performance gap between the fully centralized method based on DDPG and the DRL-based beamforming with distributed experience (i.e. D4PG). The performance of the DRL-based beamforming with distributed learning is not affected by the increase in the size of the neural networks. The reason is that, in distributed learning, the number of the optimization variables per eAP (i.e. $K$ elements in a row of the beamforming matrix) is small so that the best possible performance is achievable even for relatively smaller actor-critic neural networks (e.g. $256 \times 128$).

Moreover, we compare the running times of the gradient ascent method and the proposed DRL models as a function of the problem dimension. For the DRL models, the running time is considered to be the time required to obtain the solution in the inference mode. We recall that the inference in deep learning is a feed-forward propagation for a trained neural network. In the simulations, we use the inference time of the policy network proposed in the DRL-based beamforming approach (DDPG) since its policy network architecture is identical to that of DRL-based beamforming with distributed experience (D4PG). Moreover, it is bigger than the policy network architecture of DRL-based beamforming with distributed learning (since the output layer of the centralized method is much bigger than that of the distributed approach). The learning rate of the gradient ascent algorithm to solve the optimization problem is $\alpha = 0.1$. The problem dimension is defined by the number of eAPs $M$ and the number of UEs $K$ in the network. Here we set $K = \frac{M}{3}$ and we vary $M$ in the range of $[15,\dots,150]$.

\begin{figure}[htb]
\minipage[t]{0.45\textwidth}
  \includegraphics[width=\textwidth]{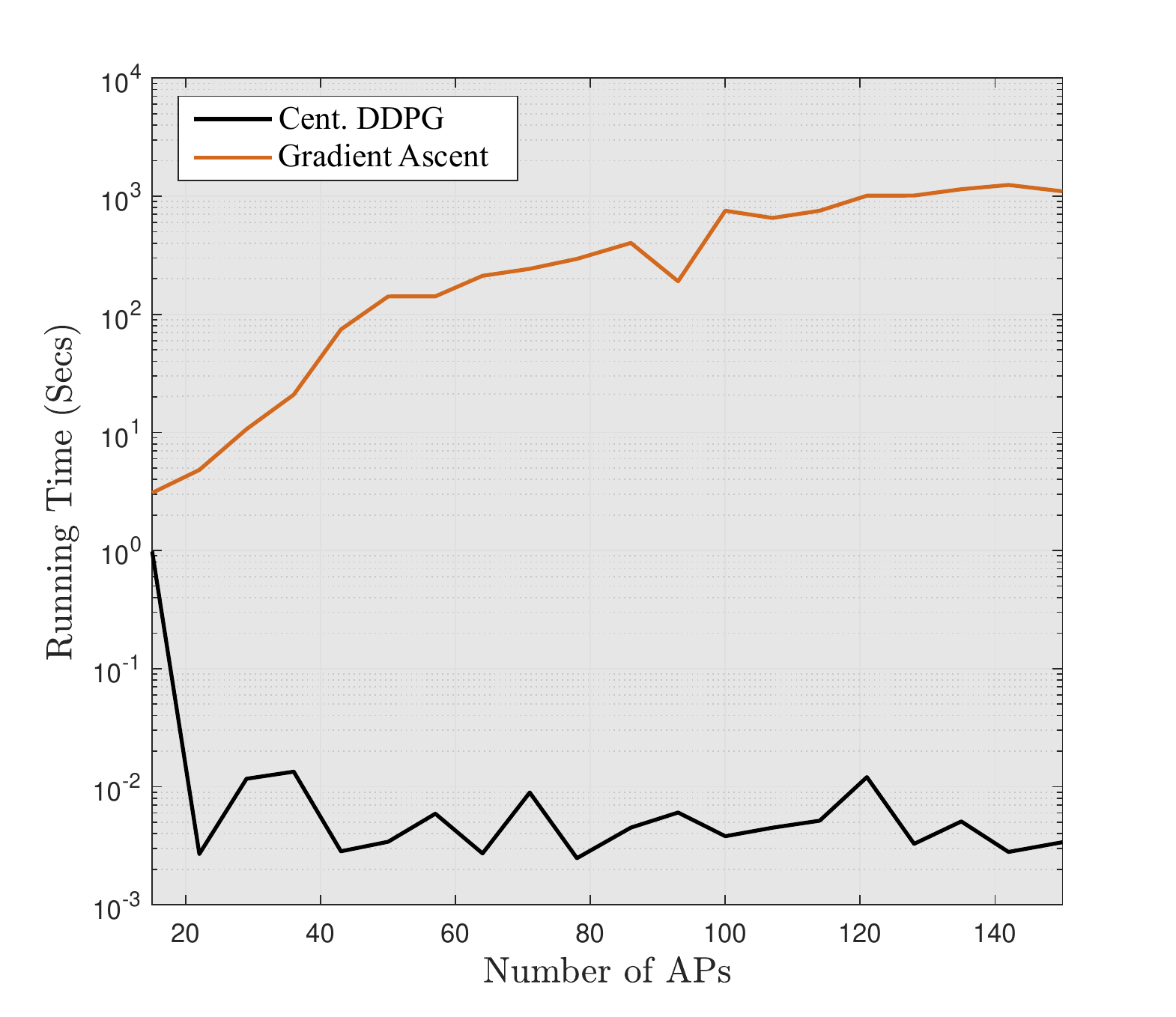}
\caption{Running time comparison.}
\label{fig:time}  
\endminipage\hfill
\minipage[t]{0.54\textwidth}
\includegraphics[width=\textwidth]{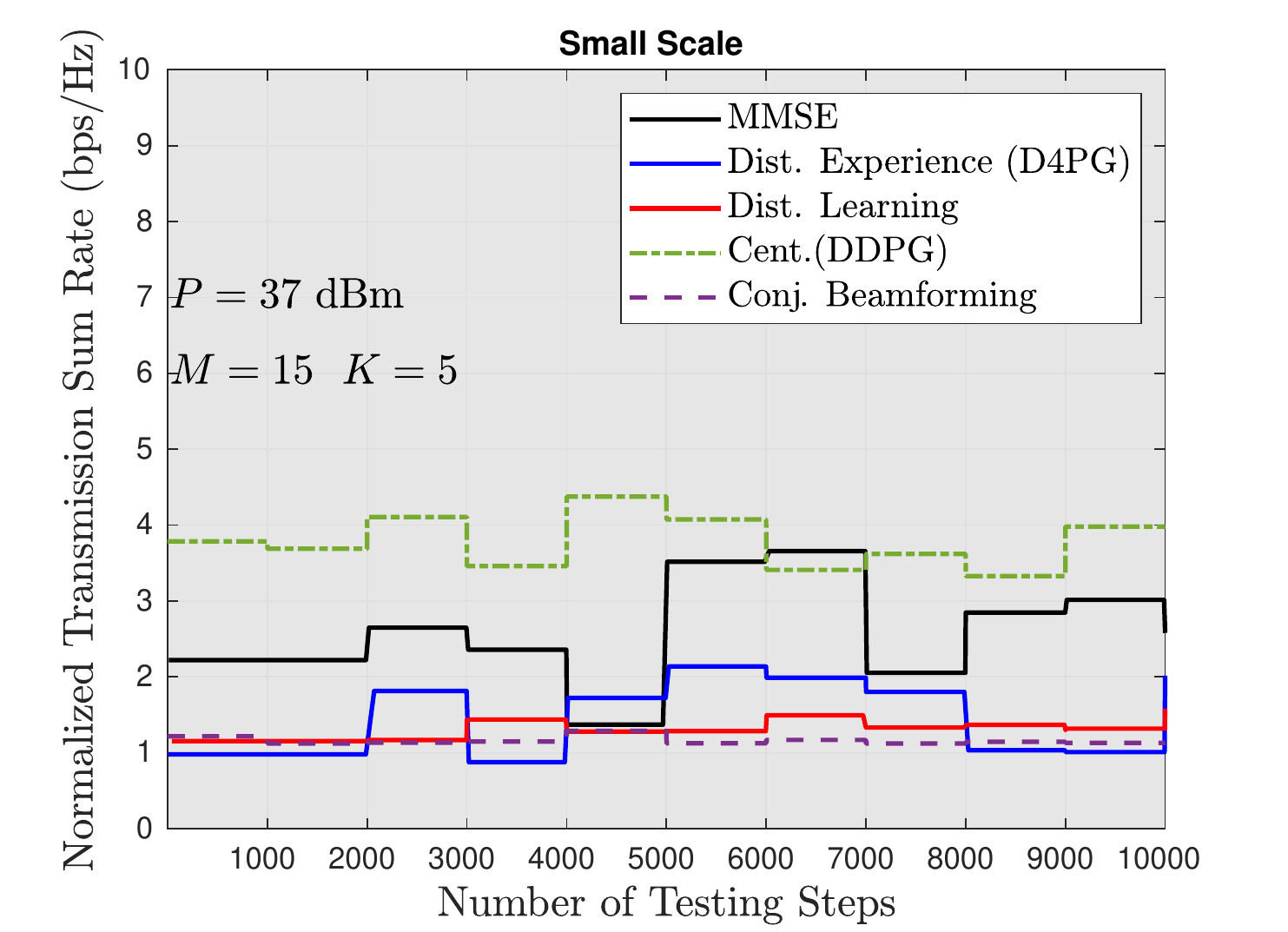}
\caption{Performance of DRL models in a testing environment.}
\label{fig:test}  
\endminipage\hfill
\end{figure}

%-------

In Fig. 8, we benchmark the gradient ascent method with the centralized DRL method in the inference mode. Note that, based on the calculation of the number of FLOPs in the inference mode, as given in Table IV, the centralized DDPG has the same complexity as the D4PG model and has higher complexity than the distributed DDPG model. That is, among the three DRL models, the centralized model involves  the most computational operations. Also, we compare the  convergence times (i.e. time needed to reach a stable solution) of the algorithms for different network size.

{Fig. \ref{fig:time}} shows that the DRL-models converge to a stable solution faster than the gradient ascent algorithm. The solution time of the gradient ascent algorithm grows exponentially with the dimension of the network, whereas the DRL models require less than a second for finding the optimal beamforming matrix. 

{Finally, we consider the case of various coherent blocks (different CSI realizations) and we study the performance of the DRL-based methods. We setup a cell-free network environment where each episode composed of 1000 iterations uses one CSI realization. We train the DRL-models on 50 episodes, 50000 iterations in totals. For benchmarking, we test the trained policies in a test environment for 10 episodes, 10 different CSI realizations, and we compare them with the centralized MMSE method and the conjugate beamforming method. Fig \ref{fig:test} shows that the MMSE method retains its superiority over the DRL-based methods, while the DRL methods outperform the conjugate beamforming method. Since the training can be done {\em offline}, the DRL methods are suitable for practical implementation due to their low inference time. This time is almost negligible compared to the convergence time of MMSE and conjugate beamforming methods, especially for huge number of eAPs and UEs.
A DRL-based method with centralized learning gives higher performance compared to the DRL-based method with distributed learning. For the latter, we split the computational task, which is costly for large cell-free networks, among the eAPs. Therefore, we have the classical trade-off between performance and computational requirements for convergence where the centralized approaches have higher performance but at the cost of  high computation at the ECP.}

%===========================================Section========================================================
\section{Conclusion}
\label{sec:Conclusion}
We have studied the beamforming optimization problem in cell-free networks. First, we have considered a fully centralized network and designed a DRL-based beamforming method based on the DDPG algorithm with continuous optimization space. We have also enhanced this method in D4PG by collecting distributed experiences from geographically-distributed eAPs. Afterward, we have developed a DRL-based beamforming design with distributed learning, which divides the beamforming optimization tasks among the eAPs. Even though the D4PG beamforming technique demonstrates a promising performance, it still conducts the learning process at the ECP. A future research direction could be to develop DRL models for channel estimation  and pilot assignment for cell-free networks, and also investigate the robustness of the DRL solutions in presence of estimation errors as well as errors in reward signals.

%%=============Appendices==============================================

\appendices
\numberwithin{equation}{section}
\section*{Appendix}
\renewcommand{\theequation}{A\thesection.\arabic{equation}}
\setcounter{equation}{0}
Let $\tau_p\leq \tau_c$, where $\tau_c$ is the coherence time of the channel via which the sequence is sent to all the eAPs with constant power. The received pilot vector at the $m$-th eAP yields
\begin{equation}
\bm{y}_{\textit{p}, m} = \sum_{k=1}^{K}\sqrt{\tau_p \rho_k}g_{ mk}\bm{\varphi}_{k} + \bm{\eta}_{m},\label{Training_1}
\end{equation}
where $\rho_k$ is the normalized transmission power for each symbol of the $k$-th UE pilot vector. Moreover, $\bm{\eta}_{m}\in \mathbb{C}^{\tau_p\times 1}$ is the zero-mean complex additive white Gaussian noise (AWGN) vector related to pilot symbols with independent and identically distributed (i.i.d) rvs, i.e. ${\eta}_{m}\thicksim \mathcal{CN}\left(0, 1/2\right)$.
To find the best estimate of $g_{mk}$ (denoted by $\hat{g}_{mk}=\mathcal{F}_{mk}^{1/2}\hat{h}_{mk}$) given the vector of observations $\bm{y}_{\textit{p}, m}$, we first project $\bm{y}_{p,m}$ over $\bm{\varphi}^H_k$. Therefore,
\begin{dmath}
\dot{y}_{p,m}=\bm{\varphi}^H_{k}\bm{y}_{\textit{p}, m}
=\underbrace{\sqrt{\tau_p \rho_k}g_{ mk}}_{\text{Desired Value}} +\underbrace{\sum_{l=1, l\neq k}^{K}\sqrt{\tau_p \rho_{l}}g_{ ml}\bm{\varphi}_{k}^H\bm{\varphi}_{l}+\bm{\varphi}_{k}^H\bm{\eta}_{m}}_{\text{Estimation Error}}.\label{Training_2}
\end{dmath}
Here $g_{mk}$ can be estimated from (\ref{Training_2}) by using the maximum {\it a posteriori} decision rule (MAP), which is identical to the minimum mean square method (MMSE) \cite{Statistical_Signal_Processing, MMSE_1}. Furthermore, given that the pilot signals are partially orthogonal and partially non-orthogonal, $\bm{\varphi}_k^H\bm{y}_{p,m}$ in (\ref{Training_2}) represents a sufficient statistics for the optimal estimation of $g_{mk}$ (MMSE). Thus the best estimate of $g_{mk}$ is given by \cite{Cell_Less_2}

\begin{equation}
\hat{g}_{mk} = \frac{\mathbb{E}\left[\dot{y}_{p,m}^*g_{mk}\right]}{\mathbb{E}\left[|\dot{y}_{p,m}|^2\right]}\dot{y}_{p,m} = \mathcal{E}_{mk}\dot{y}_{p,m}.\label{Estimate_Formula}
\end{equation}

Under the assumption that for all $~m$ and $k$, $g_{mk}$s are proper independent but non-identically distributed (i.n.d) complex Gaussian rvs, and that $\bm{\eta}_{m}$s are zero-mean i.i.d random variables, we get $\mathcal{E}_{mk}$ as shown in \textbf{Lemma \ref{Lemma_1}}.

%=============Bibliography==============================================
%\hspace{3mm}
\bibliographystyle{IEEEtran}
\bibliography{IEEEabrv,yasser2.bib}

\end{document}